\documentclass[10pt,journal]{IEEEtran}
\IEEEoverridecommandlockouts
\usepackage[utf8]{inputenc}

\usepackage{amsmath,graphicx}
\usepackage{amsmath}
\usepackage{amsthm}
\usepackage{amssymb}
\usepackage{mathrsfs}
\usepackage{mathtools}
\usepackage{bm}
\usepackage[numbers,sort&compress]{natbib}
\usepackage{graphics}
\usepackage{graphicx}
\usepackage{subfigure} 
\usepackage{acronym}
\usepackage{color}
\usepackage{comment}
\usepackage{algorithm}
\usepackage{algorithmic}
\usepackage{flushend}
\newtheorem{theorem}{Theorem}

\newtheorem{lemma}[theorem]{Lemma}

\usepackage{acronym}

\title{Approximating Multi-Dimensional and Multiband Signals}
\author{Yuhan Li, Tianyao Huang$^*$, Yimin Liu and Xiqin Wang\vspace{-2em}
	\thanks{This work was supported by the National Natural Science Foundation of China under Grants 62171259. 
	T. Huang is with the  School of Computer and Communication Engineering, University of Science and Technology Beijing, Beijing,  China. Y. Liu, Xiqin Wang, and Y. Li are with the EE Department, Tsinghua University, Beijing, China.
		$^*$Correspondence: huangtianyao@ustb.edu.cn.
}}

\acrodef{dpss}[DPSSs]{discrete prolate spheroidal sequences}
\acrodef{pswf}[PSWFs]{prolate spheroidal wave functions}
\acrodef{mse}[MSE]{mean squared error}
\acrodef{1d}[1-D]{1-dimensional}
\acrodef{2d}[2-D]{2-dimensional}
\acrodef{nd}[n-D]{$n$-dimensional}

\begin{document}
\maketitle

\begin{abstract}
We study the problem of representing a discrete tensor that comes from finite uniform samplings of a multi-dimensional and multiband analog signal. Particularly, we consider two typical cases in which the shape of the subbands is cubic or parallelepipedic. For the cubic case, by examining the spectrum of its corresponding time- and band-limited operators, we obtain a low-dimensional optimal dictionary to represent the original tensor. We further prove that the optimal dictionary can be approximated by the famous \ac{dpss} with certain modulation, leading to an efficient constructing method. For the parallelepipedic case, we show that there also exists a low-dimensional dictionary to represent the original tensor. We present rigorous proof that the numbers of atoms in both dictionaries are approximately equal to the dot of the total number of samplings and the total volume of the subbands. Our derivations are mainly focused on the \ac{2d} scenarios but can be naturally extended to high dimensions.

\end{abstract}

\begin{IEEEkeywords}
Multi-dimensional signals, multiband signals, discrete prolate spheroidal sequences, eigenvalues
	\end{IEEEkeywords}

 \section{Introduction}

A fundamental problem in signal processing is to represent or approximate the signal with a few atoms from a certain dictionary \cite{zhu2017approximating}. Such a representation or approximation benefits many tasks including reconstruction \cite{eldar2009robust}, denoising \cite{donoho1995noising}, and classification \cite{4483511}. The Fourier basis is a prevalent choice for signal representation, as it often captures signals efficiently due to the majority of their Fourier coefficients being negligible or null \cite{indyk2014sample}. However, the Fourier basis is not always the optimal choice for some natural signals, especially the multiband signal, encompassing multiple bands with non-negligible bandwidth.
 
Multiband signals have been widely concerned in communication \cite{mishali2010from,wakin2012nonuniform} and radar \cite{cohen2016carrier,zhu2015wall} applications.
The multiband signals possess a special structure in the frequency domain, whose spectral support set is constrained to certain continuous intervals \cite{mishali2009blind,mishali2010from}, that is, there is no energy or information outside these frequency intervals. These intervals are called bands. Besides, when there is only one interval, the multiband signals degrade to a band-limited signal.

For band-limited signals, some works have managed to achieve optimal approximation using \ac{pswf} \cite{slepian1961prolate,landau1961prolate,landau1962prolate} for in the continuous scenarios and \ac{dpss} \cite{slepian1978prolate,zemen2005time,davenport2012compressive} as well as periodic discrete prolate spheroidal sequences (PDPSSs) \cite{zhu2019eigenvalue,jain1981extra} in the discrete scenarios. These functions stem from the eigenvectors of the corresponding time- and band-limiting operators or covariance matrices \cite{zhu2019eigenvalue}. The distribution of their eigenvalues performs a sharp phase transition from 1 to 0, that is, most eigenvalues cluster to 1 or 0. The number of the non-zero eigenvalues is roughly equal to the product of the number of samplings and the proportion of the supporting set. Such a property leads to a natural low-dimension dictionary or basis to approximate the original band-limited signal and thus inspires wide applications \cite{zhu2017approximating}. The analysis of DPSSs is extended to multiband signals in \cite{zhu2017approximating}, which shows that the multiband signals hold similar distributions of eigenvalues as the band-limited signals. The low-dimension dictionary for the multiband signals can be approximated by the dictionaries from different band-limited signals with certain frequency modulations.

Yet, these works focus on the \ac{1d} signal and lack discussions on multi-dimensional signals. Here, we mean the domain of the signals is ``multi-dimensional" rather than the range, e.g. a vector is \ac{1d} and a matrix is \ac{2d} in the context of our framework. The multi-dimensional signals often occur in array signal processing, image processing, and other practical scenarios \cite{mishali2010from, xu2023oblivious}. These signals are distributed across multi-dimensional frequency domains, and exhibit complicated coupling relationships among these domains, rendering their analysis notably challenging. Most works \cite{slepian1964prolate,ramani2006non} consider a specific distribution and tend to find a near-optimal sampling method. The work \cite{xu2023oblivious} extends the results from \cite{avron2019leverage} to obtain an oblivious sampling method for the Fourier-constrained multi-dimensional signal. Such a method may fall sub-optimal for a specially structured signal like the multiband signal.

In this paper, we focus on the multi-dimensional multiband signals. Particularly, we mainly consider two typical cases where the uni-band is cubic or parallelepipedic, respectively. The first cubic case is prevalent in radar signal processing scenarios \cite{zhu2015wall, jonathan2018subspace}, wherein the processing typically involves discretizing the signal domain into grid points. Each grid point encapsulates the signals within its spatial cubic block. The second case is encountered in spatiotemporal adaptive processing \cite{li2022bistatic}. We primarily discuss scenarios in the more common two-dimensional setting, with the conclusion naturally extending to higher dimensions.

 Our contribution can be concluded as follows
 \begin{itemize}
     \item  For the multi-dimensional multiband signal whose subbands are cubic, we investigate its ``covariance matrix", analyze the distribution of its eigenvalues, and give the number of atoms enough to approximate the original signal in the sense of \ac{mse}. Strict proofs are provided.
     \item We show that such a multi-dimensional multiband dictionary formed by the optimal atoms can be approximated with modulated base-band DPSSs. The base-band ones can be constructed more easily, enhancing the practical value.
     \item  For the multi-dimensional multiband signal whose subbands are parallelepipeds, we theoretically analyze its ``covariance matrix" and show it possesses clustering property in the distribution of eigenvalues, which is almost the same as those with cubic subbands.
 \end{itemize}
 
We remark that our work is greatly inspired by \cite{zhu2017approximating}, which focuses on the unidimensional multiband signal. However, it is not trivial to extrapolate their findings to multi-dimensional scenarios as the complicated geometry of the subbands should be considered. This extension is challenging not only because it increases the computational complexity within the proof process, but also because it addresses the inter-dimensional coupling relationships for parallelepipedic subbands.

The remaining chapters are arranged as follows. In Section~\ref{sec:1dpre}, we introduce the model of one-dimensional multiband signals and existing theoretical results. Then in Section~\ref{sec:2dpre}, we extend the one-dimensional model to multi-dimensional multiband signals and present the notations. Section~\ref{sec:multibandlimited} provides our theoretical results concerning the properties of multi-dimensional multiband signals. The final section summarizes the entire paper.

In this paper, boldface letters like $\bm{x}$ and $\bm{B}$, except $\bm{W}$, are employed to denote tensors (including vectors and matrices). We utilize $\bm{W}$, $\mathbb{W}$, and $\mathbb{W}^{\Diamond}$, along with their superscript or subscript variations, to indicate regions in the Fourier domain. The denotations $\|\bm{W}\|$, $\|\mathbb{W}\|$, and $\|\mathbb{W}^{\Diamond}\|$ represent the Lebesgue measure of the regions. The Frobenius norm of a tensor, denoted as $\|\cdot\|_F$, is indeed defined as the square root of the sum of the squares of all its elements.



 \section{Preliminaries on \ac{1d} Multiband Signal}\label{sec:1dpre}

In this part, we first present the mathematical formulation of  \ac{1d} multiband signals. Subsequently, in order from simple to complex, we briefly review some important properties of band-limited, band-pass, and multiband signals from \cite{zhu2017approximating}. The band-limited and band-pass signals can be considered as special, simpler forms of multiband signals.

To begin with, given an \ac{1d} multiband analog signal $x(t)$, whose supporting set on the frequency domain is defined as 
\begin{equation}
    \mathbb{F} = \bigcup_{i=0}^{J-1}[F_i - B_i/2, F_i+B_i/2].
\end{equation}
Uniformly sampling the signal in the time interval $[0,NT_s)$ with sampling period $T_s$, we obtain a length-$N$ vector $\bm{x}$. Here, we assume that the sampling period $T_s$ satisfies the minimum Nyquist sampling rate. The sampled vector $\bm{x}$ also has a multiband structure as 
\begin{equation}
    \bm{x}[n] = \int_{\mathbb{W}} \tilde{x}(f) e^{j2\pi fn} \text{d}f, n = 0,1,\dots,N-1,
\end{equation}
where  the supporting set $\mathbb{W}$ is defined as 
\begin{align}
    \mathbb{W} = T_s\mathbb{F} &= \bigcup_{i=0}^{J-1}[T_s F_i - T_s B_i/2, F_i+T_s B_i/2] \notag 
    \\ 
    &= \bigcup_{i=0}^{J-1}[f_j - W_j, f_j +W_j] \subset [-\frac{1}{2},\frac{1}{2}]. 
\end{align}
The weight function $\tilde{x}(f) = \frac{1}{T_s} X(\frac{f}{T_s}), |f| \le \frac{1}{2}$, where $X()$ is the expression of the Fourier transform of $x(t)$. 

Existing works \cite{zhu2017approximating,zhu2015wall} show that such finite-length multiband vector $\bm{x}$ has a certain structure that it can be approximated by a linear combination of entries in a low-dimensional dictionary. It is assumed that $x(t)$ is a stationary stochastic signal and its power spectrum after sampling is 
\begin{equation}
    S_{xx}(f) = \left\{
    \begin{aligned}
        &1,\ \ f \in \mathbb{W},\\
        &0,\ \ {\text{else}}.
    \end{aligned}
    \right.
\end{equation}
The covariance matrix $\bm{B}_{N,\mathbb{W}}$ of $\bm{x}$ is calculated as
\begin{align}
    \bm{B}_{N,\mathbb{W}}[m,n] &=  \int_{-\frac{1}{2}}^{\frac{1}{2}} S_{xx}(f) e^{j2\pi f(m-n)} \text{d}f
        \notag \\ 
    &= \int_{\mathbb{W}} e^{j2\pi f(m-n)} \text{d}f
    \notag \\ 
    &= \sum_{i=0}^{J-1} e^{j2\pi f_i (m-n)} \frac{\sin(2\pi W_i(m-n))}{\pi(m-n)}.\label{eq:1dcovar}
\end{align}

A natural method to approximate $\bm{x}$ is performing principal component analysis on the covariance matrix $\bm{B}_{N,\mathbb{W}}$. 

\subsection{Band-limited Signal}
Before dealing with the relatively complicated matrix $\bm{B}_{N,\mathbb{W}}$, we consider a simple case where $\mathbb{W} = [-W,W], 0<W<\frac{1}{2}$, i.e. the signal $x(t)$ is a band-limited signal. We use the symbol $\bm{B}_{N,W}$ to represent the covariance matrix in this simplified case. 
The eigendecomposition of matrix $\bm{B}_{N,W}$ is
\begin{equation}
    \bm{B}_{N,W} = \bm{S}_{N,W}\bm{\Lambda}_{N,W}\bm{S}^{H}_{N,W},
\end{equation}
where the $l$-th column of $\bm{S}_{N,W}$ is the well-known DPSS vector $\bm{s}^{(l)}_{N,W}$, and $\bm{\Lambda}_{N,W}$ is a diagonal matrix with the $l$-th diagonal element being the $l$-th eigenvalue $\lambda_{N,W}^{(l)}$. The eigenvalues have a clustering property as the following lemma states.

\begin{lemma}(\cite{davenport2012compressive, slepian1978prolate})\label{lemma:1Dcluster}
1. Fixing $\epsilon \in (0,1)$, there exist constants $C_1(W,\epsilon)$, $C_2(W,\epsilon)$ and integers $N_1(W,\epsilon)$ that
\begin{equation}
\begin{aligned}
    1 - \lambda_{N,W}^{(l)} \le C_1(W,\epsilon)e^{-C_2(W,\epsilon)N} , \\ \forall l\le \lfloor 2NW (1-\epsilon) \rfloor
\end{aligned}
\end{equation}
for all $N \ge N_1(W, \epsilon)$.

2. Fixing $\epsilon \in (0,\frac{1}{2W}-1)$, there exist constants $C_3(W,\epsilon)$, $C_4(W,\epsilon)$ and integers $N_2(W,\epsilon)$ that
\begin{equation}
\begin{aligned}
     \lambda_{N,W}^{(l)} \le C_3(W,\epsilon)e^{-C_4(W,\epsilon)N} , \\ \forall l\ge \lceil 2NW (1+\epsilon) \rceil
\end{aligned}
\end{equation}
for all $N \ge N_2(W, \epsilon)$.
\end{lemma}
\noindent 
This lemma shows that there are approximately $2NW$ eigenvalues close to 1 and the rest are close to 0. Consequently, one can use the first $2NW$ eigenvectors to construct a sparse approximate representation of $\bm{x}$. 

\subsection{Band-pass Signal}
Further, we consider the band-pass signal whose supporting set in the frequency domain is $[f_c - W, f_c + W]$. The covariance matrix $\bm{B}_{N,[-W+f_c,W+f_c]}$ of such a signal can be viewed as $\bm{B}_{N,W}$ with a frequency modulation, that is $\bm{B}_{N,[-W+f_c,W+f_c]} = \bm{E}_{f_c}\bm{B}_{N,W}\bm{E}_{f_c}^{H}$. The modulation factor $\bm{E}_{f_c}$ is a diagonal matrix defined as
\begin{equation}
    \bm{E}_{f_c}[m,n] = \left\{
    \begin{aligned}
        &e^{j2\pi f_c m}, \ \ m=n, \\
        &0,\ \ m \neq n.
    \end{aligned}
    \right.
\end{equation}
From the special structure, it can be verified that
\begin{equation}
    \bm{B}_{N,[-W+f_c,W+f_c]} = \bm{E}_{f_c}\bm{S}_{N,W}\bm{\Lambda}_{N,W}(\bm{E}_{f_c}\bm{S}_{N,W})^{H}.
\end{equation}
Therefore, $(\lambda_{N,W}^{(l)},\bm{E}_{f_c}\bm{s}_{N,W}^{(l)})$ is a pair of eigenvalue and eigenvector of  $ \bm{B}_{N,[-W+f_c,W+f_c]}$. Thus, its eigenvalues have the same clustering property as Lemma~\ref{lemma:1Dcluster} states.

\subsection{Multiband Signal}
Though the above conclusions depict the clustering property of eigenvalues and show the way to approximate band-limited and band-pass signals, the extension to multiband signals is still hard. Based on a series of mathematical techniques, \cite{zhu2017approximating} proves that  $\bm{B}_{N,\mathbb{W}}$ also has a cluttering phenomenon in its distribution of eigenvalues. The result is recited as follows.
\begin{theorem}(\cite[Theorem 3.4]{zhu2017approximating})\label{lemma:1Dmulticluster}
Let $\mathbb{W} \subset [-\frac{1}{2},\frac{1}{2}]$ be formed by $J$ disjoint intervals.
1. Fix $\epsilon \in (0,1)$, there exist constants $\hat{C}_1(\mathbb{W},\epsilon)$, $\hat{C}_2(\mathbb{W},\epsilon)$ and integers $\hat{N}_1(\mathbb{W},\epsilon)$ that
\begin{equation}
\begin{aligned}
    1 - \lambda_{N,\mathbb{W}}^{(l)} \le \hat{C}_1(\mathbb{W},\epsilon)N^2 e^{-\hat{C}_2(\mathbb{W},\epsilon)N} , \\ \forall l\le J-1+\sum_i\lfloor  2NW_i (1-\epsilon) \rfloor
\end{aligned}
\end{equation}
for all $N \ge \hat{N}_1(\mathbb{W}, \epsilon)$.

2. Fix $\epsilon \in (0,\frac{1}{|\mathbb{W}|}-1)$, there exist constants $\hat{C}_3(\mathbb{W},\epsilon)$, $\hat{C}_4(\mathbb{W},\epsilon)$ and integers $\hat{N}_2(\mathbb{W},\epsilon)$ that
\begin{equation}
\begin{aligned}
    \lambda_{N,\mathbb{W}}^{(l)} \le \hat{C}_3(\mathbb{W},\epsilon) e^{-\hat{C}_4(\mathbb{W},\epsilon)N} , \\ \forall l\ge\sum_i\lceil  2NW_i (1+\epsilon) \rceil
\end{aligned}
\end{equation}
for all $N \ge \hat{N}_2(\mathbb{W}, \epsilon)$.
\end{theorem}

Figure~\ref{fig:1dillu} offers an example of the above theorems. We consider an \ac{1d} multiband signal whose power spectrum is shown in Figure ~\ref{fig:1dspec}. The power spectrum density is 1 in intervals $[-0.15,-0.05]$ and $[0.15,0.25]$, and is 0 in other regions. Thus, the total measure $\|\mathbb{W}\|$ is $0.2$. Sample $N = 256$ points evenly at intervals of 1 on this signal, and we can compute the covariance matrix of the 256 samplings as \eqref{eq:1dcovar}. Figure~\ref{fig:1deigen} presents the distribution of the eigenvalues of the covariance matrix. As stated in Theorem~\ref{lemma:1Dmulticluster}, about $N\|\mathbb{W}\|$ eigenvectors have eigenvalues near 1 while the others have eigenvalues near 0.

\begin{figure}\label{fig:1dillu}
    \centering
    \subfigure[Power spectrum]{\includegraphics[width = 0.85\linewidth]{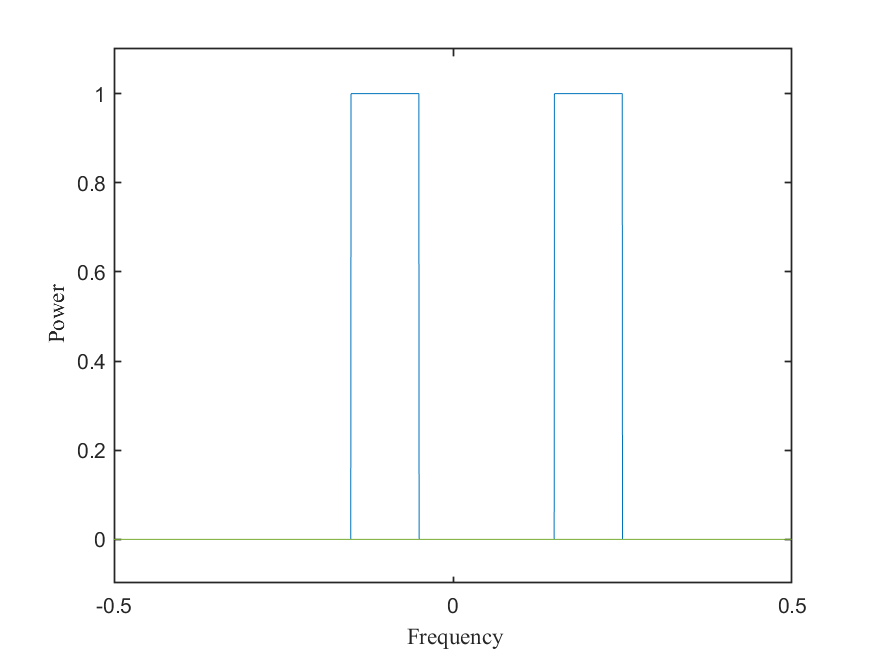}  }\label{fig:1dspec}
    \subfigure[The distribution of the eigenvalues]{\includegraphics[width = 0.85\linewidth]{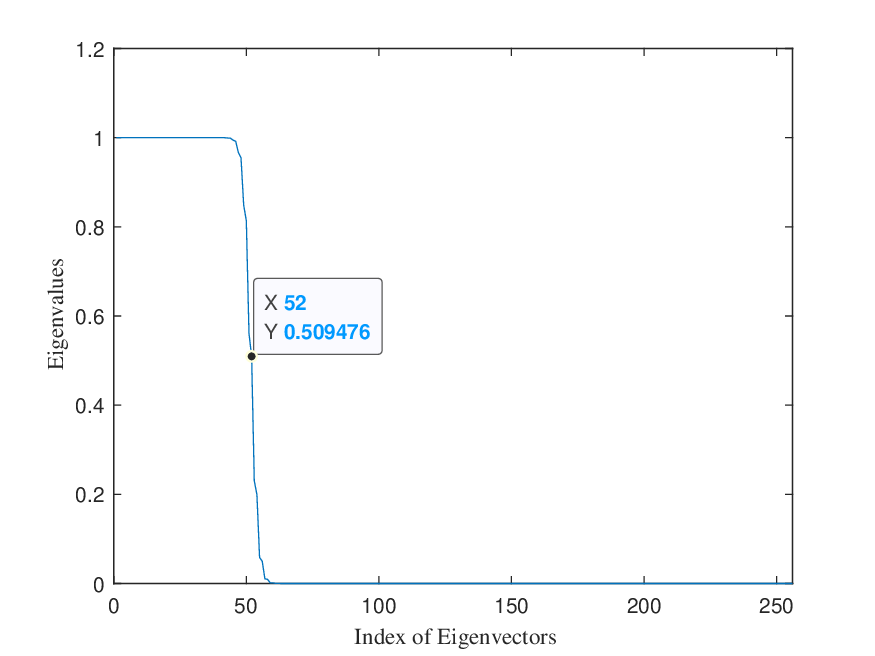}  } \label{fig:1deigen}
    \caption{An example of \ac{1d} multiband signal, whose energy is located in two disjoint intervals on the frequency range $[=0.15,-0.05]$ and $[0.15,0.25]$.}
\end{figure}

Similarly, this theorem indicates that about $N\|\mathbb{W}\|$ fixed vectors is enough to give an approximation for multiband signals whose supporting set is $\mathbb{W}$.

 \section{Multi-dimensional Multiband Signal}\label{sec:2dpre}
 In this part, we formulate our problem in multi-dimensional cases, paving the way for our theoretical analysis in the following sections. Different from \ac{1d} cases, the subbands in multi-dimensional cases have different shapes, among which we concentrate on the cubic and parallelepipedic scenarios. Different shapes with multi-dimensions increase the difficulty of analysis and symbolic expression.
 Therefore, instead of directly investigating the covariance matrix, we define a band and index-limited linear operator following \cite{zhu2017approximating} and consider the eigenvalues of the operator, which are the same as those of the covariance matrix. 

We firstly focus on the $d$-dimensional signals $x(\bm{t})$ with cubic subbands, which can be formulated by 
\begin{equation}
    x(\bm{t}) = \int_{\mathbb{F}} X(\bm{F}) e^{j2\pi \left<\bm{F}, \bm{t}\right>} \text{d}\bm{F},
\end{equation}
where the frequency vector $\bm{F}$ and the union of $J$ bands $\mathbb{F}$ is given by
\begin{align}
    \bm{F} &= \left[F^{[0]}, F^{[1]}, \cdots, F^{[d-1]}\right]^{T}, \\
    \mathbb{F} &= \bigcup_{i=0}^{J-1} \left[F^{[0]}_i - B^{[0]}_i, F^{[0]}_i + B^{[0]}_i\right] 
    \times \left[F^{[1]}_i - B^{[1]}_i, F^{[1]}_i + B^{[1]}_i\right] \notag \\
    &\times \cdots \times \left[F^{[n-1]}_i - B^{[n-1]}_i, F^{[n-1]}_i + B^{[n-1]}_i\right] \notag \\
    &= \bigcup_{i=0}^{J-1} \left[\bm{F}_i - \bm{B}_i,\bm{F}_i + \bm{B}_i \right].
\end{align}

Then let $\bm{x}$ denote the vector obtained by uniformly sampling $x(\bm{t})$ which satisfies the Nyquist sampling rate in each dimension with periods $\bm{T}_s$. With $\mathcal{N}$ denoting $\bigtimes_{i=0}^{d} \{0,1,\dots,N_i -1\}$, the sampled signal $\bm{x}$ is given by 
\begin{equation}
    \bm{x}[\bm{n}] = \int_{\mathbb{W}} \tilde{x}(\bm{f}) e^{j2\pi \left<\bm{f},\bm{n}\right>} \text{d}\bm{f},  \bm{n} \in \mathcal{N},
\end{equation}
where $\tilde{x}(\bm{f})$ is the scaled CTFT of $x(\bm{t})$ and the union of sampled bands $\mathbb{W}$ is given by
\begin{equation}
    \mathbb{W} = \bm{T}_s \circ \mathbb{F}.
\end{equation}

We are interested in the dictionary of the signal $\bm{x}$, i.e., how many tensors are enough to approximate $\bm{x}$. A natural approach is to examine its covariance ``matrix" and apply PCA. Yet, directly using the covariance ``matrix" brings inconvenience in both expressions and calculating. Below, we construct an operator that has the same eigenvalues and eigen-tensors as the covariance ``matrix".

To begin with, let $\mathcal{B}_{\mathbb{W}}: \ell_{2}(\mathbb{Z}^{d}) \rightarrow \ell_{2}(\mathbb{Z}^{d})$ denote the high-dimension multiband-limiting operator. We have
\begin{equation}
\begin{aligned}
    \mathcal{B}_{\mathbb{W}}(y)[\bm{m}] &= \int_{\mathbb{W}} e^{j2\pi \left<\bm{f}, \bm{m}\right>} \text{d}\bm{f} \star y[\bm{m}]\\
    &= \sum_{\bm{n} = -\infty^{d}}^{\infty^d}\left(y[\bm{n}] \int_{\mathbb{W}} e^{j2\pi \left<\bm{f}, (\bm{m-n})\right>} \text{d}\bm{f}\right).
\end{aligned}
\end{equation}

Further define the index-limiting operator $ \mathcal{I}_{\mathcal{N}} : \ell_{2}(\mathbb{Z}^d) \rightarrow \mathbb{C}^{\times_{i=0}^{d-1} N_{i}}$ and its inverse operator $\mathcal{I}_{\mathcal{N}}^{*}$, where $\mathcal{N} := \{0,1,\cdots,N_0 - 1\} \times \{0,1,\cdots,N_1 - 1\}\times \cdots \times \{0,1,\cdots,N_{d-1} - 1\}$

\begin{equation}
    \mathcal{I}_{\mathcal{N}}(y)[\bm{m}] = y[\bm{m}], \bm{m} \in \mathcal{N}
\end{equation}

\begin{equation}
    \mathcal{I}^*_{\mathcal{N}}(\bm{y})[\bm{m}] := \begin{cases}
    &\bm{y}[\bm{m}], \bm{m} \in \mathcal{N}\\
    &0, \ \ \text{otherwise}.
    \end{cases}
\end{equation}

Combining these operators, we obtain the linear operator $\mathcal{B}_{\bm{N},\mathbb{W}}$ which makes the signal $\bm{y}$ band-limited and index-limited as
\begin{equation}\label{eq:opdef}
\begin{aligned}
    \mathcal{B}_{\bm{N}, \mathbb{W}}(\bm{y})[\bm{m}]&:= \mathcal{I}_{\bm{N}}(\mathcal{B}_{\mathbb{W}}(\mathcal{I}_{\bm{N}}^{*}(\bm{y})))[\bm{m}] \\
    &= \sum_{\bm{n} \in \mathcal{N}}\left(\bm{y}[\bm{n}] \int_{\mathbb{W}} e^{j2\pi \left<\bm{f},(\bm{m}-\bm{n})\right>  \text{d}f}\right), \bm{m} \in \mathcal{N}.
\end{aligned}
\end{equation}






In multi-dimensional cases, for brevity, we use $\tilde{W}^{[j]}_i$ denote the interval $[-W^{[j]}_i+f^{[j]}_i,W^{[j]}_i+f^{[j]}_i]$, which represent the range of the $i$-th narrow band in the $j$-th dimension.
In \ac{2d} cases where $\bm{y} \in \mathbb{C}^{M \times N}$ is a matrix or \ac{2d} tensor, one can write $\mathcal{B}_{\bm{N},\mathbb{W}}$ as follows

\begin{equation}
    \mathcal{B}_{\bm{N},\mathbb{W}}(\bm{y}):=  \sum_{i=0}^{J-1} \bm{B}_{M, \tilde{W}^{[0]}_i} \  \bm{y} \ \bm{B}^{T}_{N, \tilde{W}^{[1]}_i},
\end{equation}
or equivalently
\begin{equation} \label{eq:2dkrone}
    \mathcal{B}_{\bm{N},\mathbb{W}}(\bm{y}):=  
    \text{ivec}_{\bm{y}}\left(\left(\sum_{i=0}^{J-1} \bm{B}_{N, \tilde{W}^{[1]}_i} \otimes \bm{B}_{M, \tilde{W}^{[0]}_i}\right) \text{vec}({\bm{y}})\right).
\end{equation}
Here, the function $\text{ivec}_{\bm{y}}(\cdot)$ rearranges a vector to a tensor whose size is the same as $\bm{y}$. Then, in $d$-dimensional space, the extended $\mathcal{B}_{\bm{N},\mathbb{W}}(\bm{y})$ can be defined as 
\begin{equation}
\begin{aligned}
    \mathcal{B}_{\bm{N},\mathbb{W}}(\bm{y})&:=  \sum_{i=0}^{J-1} 
    \text{ivec}_{\bm{y}}\left(\left(\bm{B}_{N_{d-1}, \tilde{W}^{[d-1]}_i} \otimes \bm{B}_{N_{d-2}, \tilde{W}^{[d-2]}_i} \right.\right.\\
    &\left.\left.\otimes \cdots \otimes \bm{B}_{N_{0}, \tilde{W}^{[0]}_i}\right) \text{vec}({\bm{y}})\right).
\end{aligned}
\end{equation}

One can verify that $\sum_{i=0}^{J-1} 
\bm{B}_{N_{d-1}, \tilde{W}^{[d-1]}_i} \otimes \bm{B}_{N_{d-2}, \tilde{W}^{[d-2]}_i}\otimes \cdots \otimes \bm{B}_{N_{0}, \tilde{W}^{[0]}_i}$ is exactly the variance matrix of the original multi-dimensional signal after sampling and vectorization. Therefore, the time- and band-limited operator $\mathcal{B}_{\bm{N},\mathbb{W}}$ and the ``variance matrix" share the same eigenvalues and eigen-tensors.
In \ac{1d} cases, the operator $\mathcal{B}_{\bm{N},\mathbb{W}}$ degrades to the matrix form $\sum_{i=0}^{J-1}\bm{B}_{N,W_i}$, whose eigenvalues are discussed in Theorem~\ref{thm:sharp}.

One important difference between the \ac{2d} case and the \ac{1d} case lies in the fact that the uni-bands of the \ac{2d} case may have various shapes. 
In addition to the cubic case, we consider the scenario that the uni-band is parallelepipedic, which can be viewed as being obtained through a certain affine transformation of a cube. 
Similarly, we can define the union of disjoint parallelepipedic subbands $\mathbb{W}^{\Diamond} = \bigcup_i \mathbb{W}^{\Diamond}_i$ and its time- and band-limited operator $\mathcal{B}_{\bm{N},\mathbb{W}^{\Diamond}}$ as \eqref{eq:opdef}. 
\begin{equation}
\begin{aligned}
    \mathcal{B}_{\bm{N}, \mathbb{W}^{\Diamond}}(\bm{y})[\bm{m}]&:= \mathcal{I}_{\bm{N}}(\mathcal{B}_{\mathbb{W}}(\mathcal{I}_{\bm{N}}^{*}(\bm{y})))[\bm{m}] \\
    &= \sum_{\bm{n} \in \mathcal{N}}\left(\bm{y}[\bm{n}] \int_{\mathbb{W}^{\Diamond}} e^{j2\pi \left<\bm{f},(\bm{m}-\bm{n})\right>}  \text{d}f\right), \bm{m} \in \mathcal{N}.
\end{aligned}
\end{equation}
 When the original stochastic continuous signal is stationary and its power spectrum is 1 on its supporting sets, the eigenvalues of $\mathcal{B}_{\bm{N},\mathbb{W}^{\Diamond}}$ are also the eigenvalues of its ``covariance matrix". However, it cannot be unfolded to a Kronecker dot form, because there exist cross-terms in this signal across different dimensions.


 \section{Main Result}

 Our main results are separated into three parts. Firstly, we prove that the eigen-tensors of the time- and band-limited operator for the multi-dimensional band-limited signals are formed by the outer products of DPSS bases, and its eigenvalues have a clustering property in Section~\ref{sec:multibandlimited}. Then, we consider the multi-dimensional multiband signals and show the eigenvalues of its time- and band-limited operator has a clustering property in Section~\ref{sub:multiband}. Thus, one can form a low-dimensional dictionary with part of its eigen-tensors to approximate the original signal. We show that the dictionary can be approximated by the outer products of DPSS bases in Section~\ref{sub:appro}, offering an efficient way to construct a dictionary with low computational complexity. In Section~\ref{sub:eigenpara}, we examine the case of parallelepipedic subbands and analyze the distribution of eigenvalues.
We mainly concern the \ac{2d} cases, which can be easily extended to higher dimensions.

\subsection{Eigenvalues and Eigen-Tensors for Multi-dimensional Band-Limited Signals with Cubic Subbands}\label{sec:multibandlimited}

Consider the base-band case that $\mathcal{B}_{\bm{N},\bm{W}} (\bm{y}) = \bm{B}_{M, W^{[0]}} \bm{y} \bm{B}^{T}_{N,W^{[1]}}$. The following equation states the relationship between the eigenvalues and the eigentensors of \ac{2d} cases and those of \ac{1d} cases.

\begin{equation}
    \begin{aligned}
         &\bm{B}_{M, W^{[0]}} \bm{s}^{(l)}_{M,W^{[0]}} {\bm{s}^{(k)}_{N,W^{[1]}}}^{T} \bm{B}^{T}_{N,W^{[1]}}\\
         = & \bm{B}_{M, W^{[0]}} \bm{s}^{(l)}_{M,W^{[0]}} \left( \bm{B}_{N,W^{[1]}} \bm{s}^{(k)}_{N,W^{[0]}} \right)^{T} \\
         =& \lambda_{M, W^{[0]}}^{(l)}\lambda_{N, W^{[1]}}^{(k)}\bm{s}^{(l)}_{M,W^{[0]}} {\bm{s}^{(k)}_{N,W^{[1]}}}^{T}.
    \end{aligned}
\end{equation}
In the above equation, $\lambda_{M, W^{[0]}}^{(l)}$ and $\bm{s}^{(l)}_{M,W^{[0]}}$ ($\lambda_{N, W^{[1]}}^{(k)}$ and $\bm{s}^{(k)}_{N,W^{[1]}}$) are the $l$-th ($k$-th) eigenvalue and the corresponding eigentensor of the matrix $\bm{B}_{M, W^{[1]}}$ ($\bm{B}_{N, W^{[1]}}$). The equation shows that the \ac{2d} operator $\mathcal{B}_{\bm{N},\bm{W}}$ has $MN$ eigentensors whose forms are $\bm{S}_{\bm{N},\bm{W}}^{(p)} = \bm{s}^{(l)}_{M,W^{[0]}} {\bm{s}^{(k)}_{N,W^{[1]}}}^{T}$, for some $l \in {0,1, \cdots, M-1}$ and $k \in {0,1, \cdots, N-1}$. The corresponding eigenvalues $\lambda_{\bm{N},\bm{W}}^{(p)} = \lambda_{M, W^{[0]}}^{(l)}\lambda_{N, W^{[1]}}^{(k)}$.


The following theorem describes the clustering property of $\lambda_{\bm{N},\bm{W}}^{p}$.

\begin{theorem}\label{thm:2Dcluster}
1. Fixing $\epsilon \in (0,1)$, there exist constants $\bar{C}_1 (\bm{W},\epsilon)$, $\bar{C}_2 (\bm{W},\epsilon)$, and integers $\bar{M}_1(\bm{W},\epsilon)$, $\bar{N}_1(\bm{W},\epsilon)$ that
\begin{equation}
\begin{aligned}
    1 - \lambda_{\bm{N},\bm{W}}^{(p)} &\le \bar{C}_1(\bm{W},\epsilon)e^{-\bar{C}_2(\bm{W},\epsilon)\min{(M,N)}} \\
    &\forall p\le \lfloor4MNW^{[0]}W^{[1]} (1-\epsilon)\rfloor
\end{aligned}
\end{equation}
for all $N \ge \bar{M}_1(\bm{W},\epsilon), M \ge \bar{N}_1(\bm{W},\epsilon)$.

2. Fixing $\epsilon \in (0,\frac{1}{4W^{[0]}W^{[1]}}-1)$,
there exist constants $\bar{C}_3 (\bm{W},\epsilon)$, $\bar{C}_4 (\bm{W},\epsilon)$ and integers $\bar{M}_2(\bm{W},\epsilon)$, $\bar{N}_2(\bm{W},\epsilon)$ that
\begin{equation}
\begin{aligned}
   \lambda_{\bm{N},\bm{W}}^{(p)} &\le \bar{C}_3(\bm{W},\epsilon)e^{-\bar{C}_4(\bm{W},\epsilon)\min{(M,N)}} \\
    &\forall p\ge \lceil4NMW^{[0]}W^{[1]} (1+\epsilon)\rceil
\end{aligned}
\end{equation}
for all $M \ge \bar{M}_2(\bm{W},\epsilon), N \ge \bar{N}_2(\bm{W},\epsilon)$.

\end{theorem}
\begin{proof}
See Appendix~\ref{proof:2Dcluster}.
\end{proof}
The theorem indicates that about $4NMW^{[0]}W^{[1]}$ eigenvalues of $\mathcal{B}_{\bm{N},\bm{W}}$ is close to 1, while the others are close to 0.

When $\bm{W}$ is not in the baseband but has a \ac{2d} modulated frequency, we can write $\mathcal{B}_{\bm{N},[-\bm{W} + \bm{f}_c, \bm{W} + \bm{f}_c]}$ as
\begin{equation}
    \mathcal{B}_{\bm{N},[-\bm{W} + \bm{f}_c, \bm{W} + \bm{f}_c]} \bm{y} = \bm{E}_{f_c^{[0]}} \bm{B}_{M,W^{[0]}} \bm{E}_{f_c^{[0]}}^{H} \bm{y} \bm{E}_{f_c^{[1]}}^{H} \bm{B}_{N,W^{[1]}} \bm{E}_{f_c^{[1]}}.
\end{equation}
The eigentensors of the operator can be written as
\begin{equation}
    \bm{\Lambda}_{\bm{N},[-\bm{W} + \bm{f}_c, \bm{W} + \bm{f}_c]}^{(p)} = \bm{E}_{f_c^{[0]}}\bm{s}^{(l)}_{M,W^{[0]}} {\bm{E}_{f_c^{[1]}}\bm{s}^{(k)}_{N,W^{[1]}}}^{T},
\end{equation}
whose corresponding eigenvalues are the same as $\lambda_{\bm{N},\bm{W}}^{(p)}$.

\subsection{Eigenvalues for Multi-dimensional MultiBand Signals Cubic Subbands}\label{sub:multiband}
In this section, we consider a multiband-limited case
\begin{equation}
    \mathbb{W} = [\bm{f}_0 \pm \bm{W}_0] \cup [\bm{f}_1 \pm \bm{W}_1] \cup \cdots \cup [\bm{f}_{J-1} \pm \bm{W}_{J-1}].
\end{equation}
We show that with such $\mathbb{W}$, the eigenvalues of the linear operator $\mathcal{B}_{\bm{N},\mathbb{W}}$ also possess a clustering property.

Firstly, we bound the eigenvalues of $\mathcal{B}_{\bm{N},\mathbb{W}}$ roughly with the following lemma.

\begin{lemma}\label{lemma:eigenbound}
    Given any $\mathbb{W} $ and $\bm{N}$, the eigenvalues $\lambda_{\bm{N},\mathbb{W}}^{(p)})$ of $\mathcal{B}_{\bm{N},\mathbb{W}}$ satisfy
    \begin{equation}
        0 < \lambda_{\bm{N},\mathbb{W}}^{(p)} < 1,
    \end{equation}

    and 

    \begin{equation}\label{eq:sumlambda}
    \sum_{p=0}^{MN-1} \lambda_{\bm{N},\mathbb{W}}^{(p)} = \sum_j 4NM W^{[0]}_j W^{[1]}_j = NM\|\mathbb{W}\|.
\end{equation}
\end{lemma}
\begin{proof}
    See Appendix~\ref{proof:eigenbound}.
\end{proof}

The distribution of the eigenvalues shows a dramatic transition, as the following theorem shows.
\begin{theorem}\label{thm:sharp}
    Given $\epsilon \in (0,\frac{1}{2})$, the number of the eigenvalues between $\epsilon$ and $1-\epsilon$ is bounded by
    \begin{equation}
        \#\{p:\epsilon \le \lambda_{\bm{N},\mathbb{W}}^{(p)}\le 1-\epsilon\} = O(\frac{JM\log(N) + JN\log(M)}{\epsilon(1-\epsilon)}).
    \end{equation}
\end{theorem}
\begin{proof}
    See Appendix~\ref{proof:sharp}.
\end{proof}

This theorem along with Lemma~\ref{lemma:eigenbound} is almost sufficient to show that about $NM\|\mathbb{W}\|$ eigenvalues are close to 1 while others are close to 0. This indicates the effective dimensionality is about $NM\|\mathbb{W}\|$ for multiband signal. Below, we give a more delicate analysis of the distribution of the eigenvalues in \ac{2d} cases.


\begin{theorem}\label{thm:cluster0}
Let $\mathbb{W} \subset [-\frac{1}{2},\frac{1}{2}]^2$ be a fixed union of $J$ disjoint intervals. Given an $\epsilon \in (0,\frac{1}{\|\mathbb{W}\|}-1)$, there exist constants $\tilde{C}_3$ and $\tilde{C}_4$ satisfying
\begin{equation}
     \lambda^{(p)}_{\bm{N},\mathbb{W}} \le\tilde{C}_3(\mathbb{W},\epsilon)e^{-\tilde{C}_4 (\mathbb{W},\epsilon) \min{(M,N)}},
\end{equation}
for all $p \ge \lceil MN\|\mathbb{W}\|(1+\epsilon) \rceil$, $N \ge \tilde{N}_2(\mathbb{W},\epsilon)$ and $M \ge \tilde{M}_2(\mathbb{W},\epsilon)$.  
\end{theorem}
\begin{proof}
    See Appendix~\ref{proof:cluster0}.
\end{proof}

\begin{theorem}\label{thm:cluster1}
Let $\mathbb{W} \in [-\frac{1}{2},\frac{1}{2}]^2$ be a fixed union of $J$ disjoint intervals. Given a $\epsilon \in (0,1)$, there exist constants $C_5(\mathbb{W},\epsilon)$
\begin{equation}
\begin{aligned}
     \lambda^{(p)}_{\bm{N},\mathbb{W}} &\ge 1 - 6M^2 N^2 C_5(\mathbb{W},\epsilon)e^{-\tilde{C}_2(\mathbb{W},\epsilon),\min{M,N}},
\end{aligned}
\end{equation}
for all $p \le J-1+\lfloor MN\|\mathbb{W}\|(1-\epsilon)\rfloor$.
\end{theorem}
\begin{proof}
    See Appendix~\ref{ap:cluster1}.
\end{proof}

The above 2 theorems show that only about $MN\|\mathbb{W}\|$ eigen-tensors are effective in constructing the multi-dimensional multiband signals. Therefore, one can use a low-dimensional dictionary formed by these eigen-tensors to approximate the original signal with low MSE.

\subsection{Approximation of Multiband Signal with Cubic Subbands}\label{sub:appro}
Up to now, we have shown that the eigenvalues of the operator $\mathcal{B}_{\bm{N},\mathbb{W}}$ have a cluster property. Therefore, one can use a dictionary with about $MN\|\mathbb{W}\|$ eigen-tensors of  $\mathcal{B}_{\bm{N},\mathbb{W}}$ to represent the original multi-dimensional multiband signal. However, calculating the eigen-tensors of $\mathcal{B}_{\bm{N},\mathbb{W}}$ is time consuming. Below, we show that the dictionary can be approximated by a dictionary $\bm{\Psi}$ formed by modulated DPSS.

Let $p \in {1,2,\dots,MN}$, $q_i \in {1,2,\dots,N}$ and $i \in {J}$. Define 3-dimensional tensors $\bm{\Phi}$ and $\bm{\Psi}$ as 
\begin{align}
    \bm{\Phi} &:= \left[ [\bm{\Phi}^{(0)}],[\bm{\Phi}^{(1)}],\dots,[\bm{\Phi}^{(p-1)}]\right], \label{eq:phi}\\
    \bm{\Psi} &:= \left[ [\bm{\Psi}^{(0,0)}],[\bm{\Psi}^{(0,1)}],\dots,[\bm{\Psi}^{(0,q_0)}],\dots [\bm{\Psi}^{(J-1,q_{J-1})}]\right]. \label{eq:psi}
\end{align}
where $\bm{\Phi}^{(k)}$ is the $k$-th eigen-tensor of $\mathcal{B}_{\bm{N},\mathbb{W}}$ and $\bm{\Psi}^{(i,j)}$ is the $j$-th eigen-tensor of $\mathcal{B}_{\bm{N},\bm{W}_i}$.

Following \cite{zhu2017approximating}, denoted by $\mathcal{S}_{\bm{\Phi}}$ and $\mathcal{S}_{\bm{\Psi}}$ the subspace spanned by $\bm{\Phi}$ and $\bm{\Psi}$, the similarity of the dictionaries can be measured by the angle between the subspace as
\begin{equation}
    \cos(\Theta_{\mathcal{S}_{\bm{\Phi}} \mathcal{S}_{\bm{\Psi}} }) := \left\{ 
    \begin{aligned}
        &\inf_{\bm{\phi} \in \mathcal{S}_{\bm{\Phi}}, \|\bm{\phi}\|_F = 1} \|\mathcal{P}_{\bm{\Psi}} \bm{\phi} \|_F,\ \   p \le \sum_{i \in [J]} q_i,\\
        &\inf_{\bm{\psi} \in \mathcal{S}_{\bm{\Psi}}, \|\bm{\psi}\|_F = 1} \|\mathcal{P}_{\bm{\Phi}} \bm{\psi} \|_F,\ \   p \ge \sum_{i \in [J]} q_i.
    \end{aligned}
    \right.
\end{equation}
Here $\mathcal{P}$ denotes the orthogonal projection operator.
The following two theorems show that $\cos(\Theta_{\mathcal{S}_{\bm{\Phi}} \mathcal{S}_{\bm{\Psi}} })$ approximates to 1, and thus these two dictionaries are similar.

\begin{theorem}\label{thm:pro1}
    Fix $\epsilon \in (0,\min(1,\frac{1}{\|\mathbb{W}\|}-1))$. Let $p = \sum_i \lceil MN \|\bm{W}_i\| (1+\epsilon)\rceil$ and $q_i \le \lfloor MN \|\bm{W}_i\| (1-\epsilon)\rfloor$. The $p\times M \times N$ tensor $\bm{\Phi}$ and the $(\sum_i q_i)\times M \times N$ tensor $\bm{\Psi}$ are defined as \eqref{eq:phi} and \eqref{eq:psi}. Then for any matrix $\bm{\Psi}^{(i,j)} \in {\bm{\Psi}}$,
    \begin{equation}
    \begin{small}
    \begin{aligned}
        &\|\bm{\Psi}^{(i,j)} -  \mathcal{P}_{\bm{\Phi}} \bm{\Psi}^{(i,j)}\|_F^2 \le 2\tilde{C}_1(\mathbb{W},\epsilon) e^{-\tilde{C}_2(\mathbb{W},\epsilon) \min(M,N)} /\\
        &\left(1 - \tilde{C}_1(\mathbb{W},\epsilon) e^{-\tilde{C}_2(\mathbb{W},\epsilon) \min(M,N)} 
          -  \tilde{C}_3(\mathbb{W},\epsilon) e^{-\tilde{C}_4(\mathbb{W},\epsilon) \min(M,N)}  \right)^2\\ 
        &:= \chi_1,
    \end{aligned}
    \end{small}
    \end{equation}
    and
    \begin{equation}
    \begin{small}
    \begin{aligned}
       &\cos(\Theta_{\mathcal{S}_{\bm{\Phi}} \mathcal{S}_{\bm{\Psi}} }) \ge\\
       &\sqrt{1 - MN \sqrt{\chi_1} - 3MN\sqrt{\tilde{C}_1(\mathbb{W},\epsilon)} e^{-\tilde{C}_2(\mathbb{W},\epsilon) \min(M,N)/2}}/\\
    &\sqrt{1+3MN\sqrt{\tilde{C}_1(\mathbb{W},\epsilon)} e^{-\tilde{C}_2(\mathbb{W},\epsilon) \min(M,N)/2}}.\\
    \end{aligned}
    \end{small}
    \end{equation}
    
\end{theorem}
\begin{proof}
    See Appendix~\ref{proof:pro1}
\end{proof}

\begin{theorem}\label{thm:pro2}
    Fix $\epsilon \in (0,\min(1,\frac{1}{\|\mathbb{W}\|}-1))$. Let $p = \sum_i \lfloor MN \|\bm{W}_i\| (1-\epsilon)\rfloor$ and $q_i \ge \lceil MN \|\bm{W}_i\| (1+\epsilon)\rceil$. The $p\times M \times N$ tensor $\bm{\Phi}$ and the $(\sum_i q_i)\times M \times N$ tensor $\bm{\Psi}$ are defined as \eqref{eq:phi} and \eqref{eq:psi}. Then for any matrix $\bm{\Phi}^{(i)} \in {\bm{\Phi}}$,
    \begin{equation}
    \begin{aligned}
        \|\bm{\Phi}^{(i)} -  \mathcal{P}_{\bm{\Psi}} \bm{\Phi}^{(i)}\|_F^2 &\le 12M^2 N^2 C_5(\mathbb{W},\epsilon)e^{-\tilde{C}_2(\mathbb{W},\epsilon) \min{(M,N)}} \\
        &\ \ +2MN\tilde{C}_3(\mathbb{W},\epsilon)e^{-\tilde{C}_4(\mathbb{W},\epsilon)\min{(M,N)}}\\
        &:= \chi_2,
    \end{aligned}
    \end{equation}
    and
    \begin{equation}
    \begin{small}
    \begin{aligned}
       &\cos(\Theta_{\mathcal{S}_{\bm{\Phi}} \mathcal{S}_{\bm{\Psi}} }) \ge \sqrt{1 - MN\chi_2}
    \end{aligned}
    \end{small}
    \end{equation}
    
\end{theorem}
\begin{proof}
    See Appendix~\ref{proof:pro2}
\end{proof}

\subsection{Eigenvalues of Multi-Dimensional Multiband Signal with Parallelepipedic subbands}\label{sub:eigenpara}

A parallelepipedic subband can be viewed as a cubic subband with affine transformation. Just as the proof of Lemma~\ref{lemma:eigenbound}, it is not difficult to see that the eigenvalues of $\mathcal{B}_{\bm{N},\mathbb{W}^{\Diamond}}$ are restricted to [0,1].
However, the coupling relationships across dimensions posed challenges for analysis, as its covariance does not lend itself to a Kronecker product form, unlike that of a cubic subband. 
Therefore, it is rather difficult to offer a strict analysis of the distribution of its eigenvalues as Theorem~\ref{thm:cluster0} and Theorem~\ref{thm:cluster1}. Instead, we show there is a sharp phase transition in the distribution as Theorem~\ref{thm:sharp}. 
\begin{theorem}\label{thm:parasharp}
    Suppose $\mathbb{W}^{\Diamond} \subset [-\frac{1}{2},\frac{1}{2}]\times[-\frac{1}{2},\frac{1}{2}]$ is a union of $J$ disjoint regions of parallelogram. For any  $\epsilon \in (0,\frac{1}{2})$, the number of the eigenvalues of $\mathcal{B}_{\bm{N},\mathbb{W}^{\Diamond}}$ between $\epsilon$ and $1-\epsilon$ is bounded by
    \begin{equation}
        \#\{p:\epsilon \le \lambda_{\bm{N},\mathbb{W}^{\Diamond}}^{(p)}\le 1-\epsilon\} = O(\frac{JM\log(N) + JN\log(M)}{\epsilon(1-\epsilon)}).
    \end{equation}
\end{theorem}
\begin{proof}
    See Appendix~\ref{proof:para}
\end{proof}

This theorem shows that the number of eigenvalues that are not close to 1 or 0  is sufficiently small. In  Appendix~\ref{proof:para}, we also show that the sum of these eigenvalues is $MN\|\mathbb{W}^{\Diamond}\|$. It is not difficult to see that about $MN\|\mathbb{W}^{\Diamond}\|$ eigenvalues are greater than $1-\epsilon$.

 \section{Conclusion}
 In this paper, we analyze the multi-dimensional multiband signal by considering a time-frequency limiting operator. We prove that part of the eigenvalues of the operator cluster to 1 while most of the others cluster to 0. This result indicates that there exists a low-dimensional representation of the multi-dimensional multiband signal, which may contribute to the reconstruction of such signal in certain cases.

\appendices

\section{Proof of Theorem~\ref{thm:2Dcluster}}\label{proof:2Dcluster}
This can be easily adopted by Lemma~\ref{lemma:1Dcluster}.
For the first judgment, we define $\epsilon' = 1-\sqrt{1-\epsilon} \in (0,1)$, $\bar{N}_1(\bm{W},\epsilon) = N_1(W^{[1]},\epsilon')$ and $\bar{M}_1(\bm{W},\epsilon) = M_1(W^{[0]},\epsilon')$. Using Lemma~\ref{lemma:1Dcluster}, we know that \begin{align}
    \lambda_{M,W^{[0]}}^{(l)} \ge 1 - C_1(W^{[0]},\epsilon')e^{-C_2(W^{[0]},\epsilon')M}, \\
    \lambda_{N,W^{[1]}}^{(k)} \ge 1 - C_1(W^{[1]},\epsilon')e^{-C_2(W^{[1]},\epsilon')N},
\end{align}
when $l\le \lfloor2MW^{[0]} (1-\epsilon')\rfloor \text{ and }  k\le \lfloor2NW^{[1]} (1-\epsilon')\rfloor$, for $N\ge \bar{N}_1(\bm{W},\epsilon)$ and $M\ge \bar{M}_1(\bm{W},\epsilon)$  Multiplying the two inequalities, we obtain
\begin{equation}
\begin{small}
\begin{aligned}
    &\ \lambda_{M,W^{[0]}}^{(l)}\lambda_{N,W^{[1]}}^{(k)} \\
    &\ge 1- \left(C_1(W^{[0]},\epsilon')e^{-C_2(W^{[0]},\epsilon')M}+C_1(W^{[1]},\epsilon')e^{-C_2(W^{[1]},\epsilon')N}\right) \\
    &\ + C_1(W^{[0]},\epsilon')e^{-C_2(W^{[0]},\epsilon')M}C_1(W^{[1]},\epsilon')e^{-C_2(W^{[1]},\epsilon')N} \\
    &\ge 1- \left(C_1(W^{[0]},\epsilon')e^{-C_2(W^{[0]},\epsilon')M}+C_1(W^{[1]},\epsilon')e^{-C_2(W^{[1]},\epsilon')N}\right).
\end{aligned}
\end{small}
\end{equation}

Denoted by $\bar{C}_1(\bm{W},\epsilon) = 2\max(C_1(W^{[0]},\epsilon'),C_1(W^{[0]},\epsilon'))$ and $\bar{C}_2(\bm{W},\epsilon) = \min (C_2(W^{[0]},\epsilon'),C_2(W^{[1]},\epsilon'))$, one have 
\begin{equation}
\lambda_{M,W^{[0]}}^{(l)}\lambda_{N,W^{[1]}}^{(k)} \ge 1- \bar{C}_1(\bm{W},\epsilon)e^{-\bar{C}_2(\bm{W},\epsilon) \min(M,N)}.
\end{equation}

Since $\lambda_{M,W^{[0]}}^{(l)}\lambda_{N,W^{[1]}}^{(k)}$ corresponds to a unique $\lambda_{\bm{N},\bm{W}}^{(p)}$, we have compeled the proof for judgement 1.

For the second judgment, define
$\epsilon_*$  the positive root to the equation
\begin{equation}
    \left[1+\left(\frac{1}{2W^{[0]}} - 1 \right)x\right]\left[1+\left(\frac{1}{2W^{[1]}} - 1 \right)x\right] = 1+\epsilon.
\end{equation}
and
\begin{align}
    \epsilon_0 &= \epsilon_*(1/2W^{[0]} - 1),\\
    \epsilon_1 &= \epsilon_*(1/2W^{[1]} - 1).
\end{align}. 
One can verify that $\epsilon_0$ and $\epsilon_1$ are monotone increasing with $\epsilon$. For $\epsilon \in \left(0,\frac{1}{4W^{[0]}W^{[1]}}\right)$, $\epsilon_0$ and $\epsilon_1$ are in the intervals $\left(0,\frac{1}{2W^{[0]}}\right)$ and $\left(0,\frac{1}{2W^{[1]}}\right)$, repectively.
Using Lemma~\ref{lemma:1Dcluster}, we have 
\begin{align}
    \lambda_{M,W^{[0]}}^{(l)} \le C_3(W^{[0]},\epsilon_0)e^{-C_4(W^{[0]},\epsilon_0)M},\\
    \lambda_{N,W^{[1]}}^{(k)} \le C_3(W^{[1]},\epsilon_0)e^{-C_4(W^{[1]},\epsilon_0)N},
\end{align}
when $l\ge \lceil2MW^{[0]} (1+\epsilon_0)\rceil \text{ and }  k\ge \lceil2NW^{[1]} (1+\epsilon_1)\rceil$, for $M \ge \bar{M}_2(\bm{W},\epsilon) = M_2(W^{[0]},\epsilon_0)$ and $N \ge \bar{N}_2(\bm{W},\epsilon) = N_2(W^{[1]},\epsilon_1)$. Noting that all the eigenvalues of $\bm{B}_{N,W}$ are less than 1, We obtain 
\begin{equation}
\begin{small}
\begin{aligned}
    &\ \lambda_{M,W^{[0]}}^{(l)}\lambda_{N,W^{[1]}}^{(k)} \le \min(\lambda_{M,W^{[0]}}^{(l)},\lambda_{N,W^{[1]}}^{(k)}) \le \\
    & \max\left(C_3(W^{[0]},\epsilon_0)e^{-C_4(W^{[0]},\epsilon_0)M} , C_3(W^{[1]},\epsilon_1)e^{-C_4(W^{[1]},\epsilon_1)N}\right),
\end{aligned}
\end{small}
\end{equation}
for all $l\ge \lceil2MW^{[0]} (1+\epsilon_0)\rceil$ or  $k\ge \lceil2NW^{[1]} (1+\epsilon_1)\rceil$.

Denoted by $\bar{C}_3(\bm{W},\epsilon) = \max(C_3(W^{[0]},\epsilon_0), C_3(W^{[1]},\epsilon_1))$ and $\bar{C}_4(\bm{W},\epsilon) = \min(C_4(W^{[0]},\epsilon_0), C_4(W^{[1]},\epsilon_1))$, we have
\begin{equation}
\lambda_{M,W^{[0]}}^{(l)}\lambda_{N,W^{[1]}}^{(k)} \le \bar{C}_3(\bm{W},\epsilon) e^{-\bar{C}_4(\bm{W},\epsilon) \min{(M,N)}}.
\end{equation}
Noting the facts that $\lambda_{M,W^{[0]}}^{(l)}\lambda_{N,W^{[1]}}^{(k)}$ corresponds to a unique $\lambda_{\bm{N},\bm{W}}^{(p)}$ and $(1+\epsilon_0)(1+\epsilon_1) = 1+\epsilon$, we have completed the proof of judgement 2.

\section{Proof of Lemma~\ref{lemma:eigenbound}}\label{proof:eigenbound}

Given $\bm{Y} \in \mathbb{C}^{M\times N}$, then the eigenvalues of $\mathcal{B}_{\bm{N},\mathbb{W}}$ satisfy

\begin{equation}
    \min_{\bm{Y}} \frac{\left<\mathcal{B}_{\bm{N},\mathbb{W}}(\bm{Y}),\bm{Y}\right>}{\left<\bm{Y},\bm{Y}\right>} \le \lambda_{\bm{N},\mathbb{W}} \le \max_{\bm{Y}} \frac{\left<\mathcal{B}_{\bm{N},\mathbb{W}}(\bm{Y}),\bm{Y}\right>}{\left<\bm{Y},\bm{Y}\right>}.
\end{equation}

Here, the inner product is defined as $\left<\bm{A},\bm{B}\right> = trace(\bm{A}^{H} \bm{A})$. We have the following equation\vspace{-1em}

\begin{equation}
\begin{aligned}
    &\left<\mathcal{B}_{\bm{N},\mathbb{W}}(\bm{Y}),\bm{Y}\right> =
    \sum_{\bm{m} = \bm{0}}^{\bm{N}} \mathcal{B}_{\bm{N},\mathbb{W}}(\bm{Y}[\bm{m}]) \bar{\bm{Y}}[\bm{m}]\\
    &= \sum_{\bm{m} = \bm{0}}^{\bm{N}} (\sum_{\bm{n} = \bm{0}}^{\bm{N}} \int_{\mathbb{W}} e^{j2\pi \left<\bm{f},(\bm{m}-\bm{n})\right>} \text{d}\bm{f}\ \bm{Y}[\bm{n}] ) \bar{\bm{Y}}[\bm{m}]\\
    & =  \int_{\mathbb{W}} \left(\sum_{\bm{m} = \bm{0}}^{\bm{N}}e^{j2\pi \left<\bm{f},\bm{m})\right>}\bar{\bm{Y}}[\bm{m}] \right)\left(\sum_{\bm{n} = \bm{0}}^{\bm{N}}e^{-j2\pi \left<\bm{f},\bm{n})\right>}\bm{Y}[\bm{n}]\right) \text{d} \bm{f}\\
    &= \int_{\mathbb{W}} \|\mathcal{F}\bm{Y}\|^2 \text{d} \bm{f} > 0.
\end{aligned}
\end{equation}
Noting $\int_{\mathbb{W}} \|\mathcal{F}\bm{Y}\|_F^2 \text{d} \bm{f} \le \int_{[-\frac{1}{2},\frac{1}{2}] \times [-\frac{1}{2},\frac{1}{2}]} \|\mathcal{F}\bm{Y}\|_F^2 \text{d} \bm{f} \le \|\bm{Y}\|_F^2$, and the eigenvalues satisfy $0<\lambda_{\bm{N},\mathbb{W}}<1$. 
Further, we bound the sum of $\lambda_{\bm{N},\mathbb{W}}$. We use the Kronecker product form as \eqref{eq:2dkrone} and obtain
\begin{equation}
\begin{aligned}
    \sum_{p=0}^{MN-1} \lambda_{\bm{N},\mathbb{W}}^{(p)} &= \sum_{J} trace(\bm{B}_{N,W^{[1]}_j} \otimes \bm{B}_{M,W^{[0]}_j})\\
    &= \sum_j 4NM W^{[0]}_j W^{[1]}_j = NM\|\mathbb{W}\|.
\end{aligned}
\end{equation}

\section{Proof of Theorem~\ref{thm:sharp}}\label{proof:sharp}
To begin with, we calculate the F norm of $\bm{B}_{N, W^{[1]}} \otimes \bm{B}_{M, W^{[0]}}$ which is also the sum of square of the eigenvalues of $\mathcal{B}_{\bm{N},\mathbb{W}}$
\begin{equation}
\begin{aligned}
    &\|\bm{B}_{N,W^{[1]}} \otimes \bm{B}_{M,W^{[0]}}\|_F^2 = \|\bm{B}_{N,W^{[1]}}\|_F^2 \|\bm{B}_{M,W^{[0]}}\|_F^2 \\
    &\ge (2NW^{[1]} - \frac{2}{\pi^2}\frac{2N-1}{N-1} - \frac{2}{\pi^2}\log(N-1))\\
    &\quad \cdot(2MW^{[0]} - \frac{2}{\pi^2}\frac{2M-1}{M-1} - \frac{2}{\pi^2}\log(M-1))
\end{aligned}
\end{equation}
The above results yield

\begin{equation}
    \begin{aligned}
        &\|\sum_{i=0}^{J-1}\bm{B}_{N,W_i^{[1]}} \otimes \bm{B}_{M,W_i^{[0]}}\|_F^2 \ge \sum_{i=0}^{J-1} \|\bm{B}_{N,W_i^{[1]}} \otimes \bm{B}_{M,W_i^{[0]}}\|_F^2\\
        &\ge \sum_{i=0}^{J-1} (2NW_{i}^{[1]} - \frac{2}{\pi^2}\frac{2N-1}{N-1} - \frac{2}{\pi^2}\log(N-1))\\
    &\quad \quad \quad \cdot (2MW_{i}^{[0]} - \frac{2}{\pi^2}\frac{2M-1}{M-1} - \frac{2}{\pi^2}\log(M-1))\\
    & \ge MN\|\mathbb{W}\| - \frac{4MJ}{\pi^2}(3+\log(N))\\
    &\quad - \frac{4NJ}{\pi^2}(3+\log(M))\\
    &\quad  + \frac{4J}{\pi^4}(3+\log(M))(3+\log(N))
    \end{aligned}
\end{equation}

Recall the sum of $\lambda_{\bm{N},\mathbb{W}}^{(p)}$ \eqref{eq:sumlambda},
and we have the following equation 
\begin{equation}
\begin{aligned}
    &\sum_{p = 0}^{MN-1} \lambda_{\bm{N},\mathbb{W}}^{(p)}(1-\lambda_{\bm{N},\mathbb{W}}^{(p)}) \\
    &= trace(\bm{B}_{N,W^{[1]}} \otimes \bm{B}_{M,W^{[0]}}) - \|\bm{B}_{N,W^{[1]}} \otimes \bm{B}_{M,W^{[0]}}\|_{F}^2 \\
    &\le  \frac{4MJ}{\pi^2}(3+\log(N))+ \frac{4NJ}{\pi^2}(3+\log(M))
\end{aligned}
\end{equation}

Take the number of $\lambda_{\bm{N},\mathbb{W}}$ as $\#\{medium\}$, and we can bound it as 
\begin{equation}
\begin{aligned}
    &\sum_{p = 0}^{MN-1} \lambda_{\bm{N},\mathbb{W}}^{(p)}(1-\lambda_{\bm{N},\mathbb{W}}^{(p)}) \ge \sum_{\epsilon<\lambda_{\bm{N},\mathbb{W}}^{(p)}<1-\epsilon}\lambda_{\bm{N},\mathbb{W}}^{(p)}(1-\lambda_{\bm{N},\mathbb{W}}^{(p)}) \\
    &\ge \epsilon(1-\epsilon)\#\{medium\}.
\end{aligned}
\end{equation}
Therefore, we have $\#\{medium\} =  \frac{O(JM\log(N) + JN\log(M))}{\epsilon(1-\epsilon)}$.

\section{Proof of Theorem~\ref{thm:cluster0}}\label{proof:cluster0}

From \cite{horn2012matrix}, we know that the following equation holds
\begin{equation}
    \lambda^{(p)}_{\bm{N},\mathbb{W}} \le \sum_{i=0}^{J-1} \lambda_{\bm{N},\bm{W}_i}^{(l_i)}.
\end{equation}
for $l_i  \in [MN]$ and $p = \sum_{i=0}^{J-1} l_i$.

Using the clustering property of $\lambda_{\bm{N},\bm{W}_i}^{(l_i)}$ \eqref{thm:2Dcluster}, we have
\begin{equation}
\begin{aligned}
\lambda_{\bm{N},\bm{W}_i}^{(l_i)} &\le \bar{C}_3(\bm{W}_i,\epsilon)e^{-\bar{C}_4(\bm{W}_i,\epsilon)\min(M,N)},\\
&\forall l_i \ge \lceil 4MN W_i^{[0]} W_i^{[1]}(1+\epsilon)  \rceil,
\end{aligned}
\end{equation}
for all $N \ge \tilde{N}_2(\mathbb{W},\epsilon) = \max_{i} \bar{N}_2(\bm{W}_i,\epsilon)$ and $M \ge \tilde{M}_2(\mathbb{W},\epsilon) = \max_{i} \bar{M}_2(\bm{W}_i,\epsilon)$. By selecting $l_i \ge \lceil MN\|\bm{W}_i\|(1+\epsilon) \rceil$, we have
\begin{equation}
\begin{aligned}
    \ \lambda^{(p)}_{\bm{N},\mathbb{W}} &\le \sum_{i=0}^{J-1}\bar{C}_3(\bm{W}_i,\epsilon)e^{-\bar{C}_4(\bm{W}_i,\epsilon)\min(M,N)} \\
    & \le \tilde{C}_{3}(\mathbb{W},\epsilon)e^{-\tilde{C}_{4}(\mathbb{W},\epsilon)\min{(M,N)}},
\end{aligned}
\end{equation}
for all $p \ge \lceil MN\|\mathbb{W}\|(1+\epsilon) \rceil$, $N \ge \tilde{N}_2(\mathbb{W},\epsilon)$ and $M \ge \tilde{M}_2(\mathbb{W},\epsilon)$, where $\tilde{C}_{3}(\mathbb{W},\epsilon) = J \max_{i}\{\bar{C}_3(\bm{W}_i,\epsilon)\}$ and  $\tilde{C}_{4}(\mathbb{W},\epsilon) = \min_{i}\{\bar{C}_4(\bm{W}_i,\epsilon)\}$.

\section{Proof of Theorem~\ref{thm:cluster1}}\label{ap:cluster1}

\subsection{$\epsilon$-pseudo eigenvalue}
In this part, we show that $\bm{E}_{f_i^{[0]}} \bm{s}_{M,W_i^{[0]}}^{(l)}{\bm{s}_{N,W_i^{[1]}}^{(k)}}^{T}\bm{E}_{g_i^{[1]}}$ is an $\epsilon-pseudo$ eigentensor of $\mathcal{B}_{\bm{N},\mathbb{W}}$, when $[\bm{f}_i \pm \bm{W}_i]$ is one subband in $\mathbb{W}$.

For the operator $\mathcal{B}_{\bm{N},\mathbb{W}}$, we have
\begin{equation} \label{eq:pseudo}
    \begin{aligned}
        &\ \mathcal{B}_{\bm{N},\mathbb{W}}(\bm{E}_{f_i} \bm{s}_{M,W_i^{[0]}}^{(l)}{\bm{s}_{N,W_i^{[1]}}^{(k)}}^{T}\bm{E}_{g_i})\\
        &=e^{j2\pi f_i m + g_i n }\lambda^{(l)}_{M,W_i^{[0]}}\lambda^{(k)}_{N,W_i^{[1]}} \bm{s}_{M,W_i^{[0]}}^{(l)}{\bm{s}_{N,W_i^{[1]}}^{(k)}}^{T} + \\
        &\ \mathcal{B}_{\bm{N},\mathbb{W}/[\bm{f}_i \pm \bm{W}_i]}(\bm{E}_{f_i} \bm{s}_{M,W_i^{[0]}}^{(l)}{\bm{s}_{N,W_i^{[1]}}^{(k)}}^{T}\bm{E}_{g_i}).
    \end{aligned}
\end{equation}

In the following, we bound the second term $\bm{o}^{(p)} :=  \mathcal{B}_{\bm{N},\mathbb{W}/[\bm{f}_i \pm \bm{W}_i]}(\bm{E}_{f_i} \bm{s}_{M,W_i^{[0]}}^{(l)}{\bm{s}_{N,W_i^{[1]}}^{(k)}}^{T}\bm{E}_{g_i})$ as

\begin{equation}\label{eq:obound}
\begin{aligned}
    &\ \|\bm{o}^{(p)}\|_F^2 = \|\mathcal{B}_{\bm{N},\mathbb{W}/[\bm{f}_i \pm \bm{W}_i]}(\bm{E}_{f_i} \bm{s}_{M,W_i^{[0]}}^{(l)}{\bm{s}_{N,W_i^{[1]}}^{(k)}}^{T}\bm{E}_{g_i})\|_F^2 \\
    &\le \|\mathcal{B}_{\mathbb{W}/[\bm{f}_i \pm \bm{W}_i]}\left(\mathcal{I}^{*}_{\bm{N}}(\bm{E}_{f_i} \bm{s}_{M,W_i^{[0]}}^{(l)}{\bm{s}_{N,W_i^{[1]}}^{(k)}}^{T}\bm{E}_{g_i})\right)\|_F^2 \\
    &= \|\bm{s}_{M,W_i^{[0]}}^{(l)}{\bm{s}_{N,W_i^{[1]}}^{(k)}}^{T}\|_F^2 \\
    &\ \ \ - \|\mathcal{B}_{[\bm{f}_i \pm \bm{W}_i]}\left(\mathcal{I}^{*}_{\bm{N}}(\bm{E}_{f_i} \bm{s}_{M,W_i^{[0]}}^{(l)}{\bm{s}_{N,W_i^{[1]}}^{(k)}}^{T}\bm{E}_{g_i})\right)\|_F^2 \\
    &\le \|\bm{s}_{M,W_i^{[0]}}^{(l)}{\bm{s}_{N,W_i^{[1]}}^{(k)}}^{T}\|_F^2 \\ 
    &\ \ \ - \|\mathcal{B}_{\bm{N},[\bm{f}_i \pm \bm{W}_i]}(\bm{E}_{f_i} \bm{s}_{M,W_i^{[0]}}^{(l)}{\bm{s}_{N,W_i^{[1]}}^{(k)}}^{T}\bm{E}_{g_i})\|_F^2 \\
    &\le 1 - (\lambda^{(l)}_{M,W_i^{[0]}}\lambda^{(k)}_{N,W_i^{[1]}})^2 \\
    &\le 1 - \left(1- \bar{C}_1(\bm{W},\epsilon)e^{-\bar{C}_2(\bm{W},\epsilon)\min{(M,N)}}\right)^2\\
    &\le 2\bar{C}_1(\bm{W},\epsilon)e^{-\bar{C}_2(\bm{W},\epsilon)\min{(M,N)}}
\end{aligned}
\end{equation}\vspace{-3em}

\subsection{cluster near 1}
In this part, we first introduce one matrix and one lemma that depicts the correlations between different $\bm{s}_{N,W_i}$. Given $\epsilon \in (0,1)$ and $k_i = \lfloor2NW_i (1-\epsilon)\rfloor,\ \forall i \in [J]$, the dictionary matrix is defined as $\bm{S} := [\bm{S}_0, \bm{S}_1, \cdots, \bm{S}_{J-1}]$, where the columns in $\bm{S}_i \in \mathbb{C}^{N \times k_i}$ is formed by the first $k_i$ eigenvectors of $\bm{B}_{N,W_i}$, $\bm{s}_{N,W_i}^{l_i},\ l_i\in [k_i]$.

\begin{lemma}\cite{zhu2017approximating}
For any pair of distinct columns $\bm{s}_{N,W_1}^{l_1}$ and $\bm{s}_{N,W_2}^{l_2}$ in $\bm{S}$ whose corresponding eigenvalues are $\lambda_1$ and $\lambda_2$, the following inequalities hold
\begin{equation}
    |\left<\bm{s}_{N,W_1}^{l_1} ,\bm{s}_{N,W_2}^{l_2}\right>| \le 3\sqrt{1-\min{(\lambda_1,\lambda_2)}},
\end{equation}
for $N \ge \tilde{N}_1(\bm{W},\epsilon)$.

\end{lemma}

Similarly, we can construct a matrix $\bm{\Psi} := [\bm{\Psi}_0, \bm{\Psi}_1, \cdots, \bm{\Psi}_{J-1}]$. Here, the columns of $\bm{\Psi}_i \in \mathbb{C}^{MN \times k_i}$ is formed by the first $k_i$ vectorized eigentensors of $\mathcal{B}_{\bm{N},\bm{W}_i}$, $\text{vec}(\bm{s}_{\bm{N},\bm{W}_i}^{(p_i)})$, for $p_i \in [k_i]$, $k_i = \lfloor 4MNW_i^{[0]}W_{i}^{[1]}(1-\epsilon) \rfloor$ and $\epsilon\in (0,1)$.
Then, for any pair of distinct columns $\bm{\psi}_1$ and $\bm{\psi}_2$, we have
\begin{equation}\label{eq:corbound1}
\begin{aligned}
    |\left<\bm{\psi}_1, \bm{\psi}_2\right>|  &= |trace(\bm{s}_{M,W^{[0]}}^{(l_0)}{\bm{s}_{N,W^{[1]}}^{(k_0)}}^{T}\bm{s}_{N,W^{[1]}}^{(k_1)}{\bm{s}_{M,W^{[0]}}^{(l_1)}}^{T})|\\
    &= |trace({\bm{s}_{M,W^{[0]}}^{(l_1)}}^{T}\bm{s}_{M,W^{[0]}}^{(l_0)}{\bm{s}_{N,W^{[1]}}^{(k_0)}}^{T}\bm{s}_{N,W^{[1]}}^{(k_1)})|\\
    &= |{\bm{s}_{M,W^{[0]}}^{(l_1)}}^{T}\bm{s}_{M,W^{[0]}}^{(l_0)}| \cdot |{\bm{s}_{N,W^{[1]}}^{(k_0)}}^{T}\bm{s}_{N,W^{[1]}}^{(k_1)}|\\
    &\le 3\sqrt{1-\min{(\lambda_{M,W^{[0]}}^{(l_1)}},\lambda_{M,W^{[0]}}^{(l_0)},\lambda_{N,W^{[1]}}^{(k_0)},\lambda_{N,W^{[1]}}^{(k_1)})}\\
    &\le 3\sqrt{1-\min{(\lambda_{\bm{N},\bm{W}}^{(p_0)},\lambda_{\bm{N},\bm{W}}^{(p_1)})}} \\
    &\le 3\sqrt{\tilde{C}_1(\mathbb{W},\epsilon)e^{-\tilde{C}_2(\mathbb{W},\epsilon)\min{(M,N)}}},
\end{aligned}
\end{equation}
where $\tilde{C}_1(\mathbb{W},\epsilon) = \max_i {\bar{C}_1(\bm{W}_i,\epsilon)}$ and $\tilde{C}_2(\mathbb{W},\epsilon) = \min_i {\bar{C}_2(\bm{W}_i,\epsilon)}$

Noting that the  diagonal elements of the matrix $\bm{\Psi}^{H}\bm{\Psi}$ are all 1 and the nondiagonal elements are bound by the above inequalities, we can bound the spectral norm of $\bm{\Psi}^{H}\bm{\Psi}$ by Gerschgorin circle theorem as
\begin{equation}
    \|\bm{\Psi}^{H}\bm{\Psi}\|_2 \le 1 + 3MN\sqrt{\tilde{C}_1(\mathbb{W},\epsilon)e^{-\tilde{C}_2(\mathbb{W},\epsilon)\min{(M,N)}}}.
\end{equation}\vspace{-1em}

Then we obtain
\begin{equation}
\begin{aligned}
    &\ \sum_{p=0}^{J-1+\sum_{i}\lfloor MN\|\bm{W}_i\|(1-\epsilon) \rfloor} \lambda_{\bm{N},\mathbb{W}}^{(p)} \\
    &\ge trace\left(\bm{\Psi}^{H}\bm{B}_{\bm{N},\mathbb{W}}\bm{\Psi}\right)  \bigg/  \|\bm{\Psi}^{H}\bm{\Psi}\|_2 \\
    & = \left(\sum_{i=0}^{J-1} \sum_{p_i=0}^{\lfloor MN\|\bm{W}_i\|(1-\epsilon) \rfloor} \left( (\text{vec}(\bm{E}_{f_i}\bm{S}_{\bm{N},\bm{W}_i}^{(p_i)}\bm{E}_{f_i}))^{H} \right)\right.\\
    &\left. \left( \lambda_{\bm{N},\bm{W}_i}^{(p_i)} \text{vec}(\bm{E}_{f_i}\bm{S}_{\bm{N},\bm{W}_i}^{(p_i)}\bm{E}_{f_i}))  +\bm{o}_i^{(p_i)} \right)\right) \bigg/ \|\bm{\Psi}^{H}\bm{\Psi}\|_2 \\
    &\ge \left(\sum_{i=0}^{J-1} \sum_{p_i=0}^{\lfloor MN\|\bm{W}_i\|(1-\epsilon) \rfloor}  (\lambda_{\bm{N},\bm{W}_i}^{(p_i)} - \|\bm{o}_i^{p_i}\|_2)  \right) \bigg/ \|\bm{\Psi}^{H}\bm{\Psi}\|_2 \\
    & \ge \left(\sum_{i=0}^{J-1} \sum_{p_i=0}^{\lfloor MN\|\bm{W}_i\|(1-\epsilon) \rfloor}  \left(1-\bar{C}_1(\bm{W}_i,\epsilon)e^{-\bar{C}_2(\bm{W}_i,\epsilon)\min{(M,N)}} \right.\right.\\
    & \left.\left. \ \ -\sqrt{2\bar{C}_1(\bm{W},\epsilon)e^{-\bar{C}_2(\bm{W},\epsilon)\min{((M,N))}}}\right)  \right) \bigg/ 
    \\ &\ \ \left(1 + 3MN\sqrt{\tilde{C}_1(\mathbb{W},\epsilon)}e^{-\tilde{C}_2 (\mathbb{W},\epsilon)\min{M,N}}\right)\\
    &\ge \left(J + \sum_{i} \lfloor MN\|\bm{W}_i\|(1-\epsilon)\rfloor \right.\\
    &\left.\ \ -3MN C_5(\mathbb{W},\epsilon)e^{-\tilde{C}_2(\mathbb{W},\epsilon)\min{(M,N)}/2}\right)\bigg/\\
    &\ \ \left(1 + 3MN C_5(\mathbb{W},\epsilon)e^{-\tilde{C}_2(\mathbb{W},\epsilon)\min{(M,N)}/2}\right)\\
    &\ge J + \sum_{i} \lfloor MN\|\bm{W}_i\|(1-\epsilon)\rfloor \\
    &\ \ \ - 6M^2 N^2 C_5(\mathbb{W},\epsilon)e^{-\tilde{C}_2(\mathbb{W},\epsilon)\min{(M,N)}/2}.
\end{aligned}
\end{equation}
Here, $C_5(\mathbb{W},\epsilon) = \max{(\tilde{C}_1(\mathbb{W},\epsilon),\sqrt{\tilde{C}_1(\mathbb{W},\epsilon)})}$. The numbers $M$ and $N$ should satisfy $M \ge \tilde{M}_1 (\mathbb{W},\epsilon)$, $N \ge \tilde{M}_1 (\mathbb{W},\epsilon)$ and $3MN C_5(\mathbb{W},\epsilon)e^{-\tilde{C}_2(\mathbb{W},\epsilon)\min{(M,N)}/2} < 1$.

Therefore, we bound $\lambda_{\bm{N},\mathbb{W}}^{(p)}$ by
\begin{equation}
    \begin{aligned}
        \lambda_{\bm{N},\mathbb{W}}^{(p)} &= \left(\sum_{p'=0}^{J-1+\sum_{i}\lfloor MN\|\bm{W}_i\|(1-\epsilon) \rfloor} \lambda_{\bm{N},\mathbb{W}}^{(p')}\right) \\
        &\ \ \ -\left( \sum_{p'=0,p'\neq p}^{J-1+\sum_{i}\lfloor MN\|\bm{W}_i\|(1-\epsilon) \rfloor} \lambda_{\bm{N},\mathbb{W}}^{(p')}\right)\\
        &\ge \left(\sum_{p'=0}^{J-1+\sum_{i}\lfloor MN\|\bm{W}_i\|(1-\epsilon) \rfloor} \lambda_{\bm{N},\mathbb{W}}^{(p')}\right) \\
        &\ \ \ - \left( J-1+\sum_{i}\lfloor MN\|\bm{W}_i\|(1-\epsilon) \rfloor\right)\\
        &\ge 1 -  6M^2 N^2 C_5(\mathbb{W},\epsilon)e^{-\tilde{C}_2(\mathbb{W},\epsilon)\min{(M,N)}/2}.
    \end{aligned}
\end{equation}

\section{Proof of Theorem~\ref{thm:pro1}}\label{proof:pro1}
We first unfold the operator $\mathcal{B}_{\bm{N},\mathbb{W}}$ by eigenfunction expansion as
\begin{equation}
    \mathcal{B}_{\bm{N},\mathbb{W}}(\cdot) = \sum_{k=0}^{NM-1} \lambda_{\bm{N},\mathbb{W}}^{(k)} \bm{\Phi}^{(k)} <\bm{\Phi}^{(k)}, \cdot>.
\end{equation}

Fix $\epsilon$. Given $i \in [J]$ and $j \le q_i$, the $j$-th eigen-tensor of $\mathcal{B}_{\bm{N},\bm{W}_i}$, $\bm{\Psi}^{(i,j)}$ satisfies
\begin{equation}
\begin{aligned}
    \lambda_{\bm{N},\mathbb{W}}^{(k)} <\bm{\Phi}^{(k)},\bm{\Psi}^{(i,j)}> &\overset{(a)}{=} <\bm{\Phi}^{(k)}, \mathcal{B}_{\bm{N},\mathbb{W}}(\bm{\Psi}^{(i,j)})>\\ 
    &\overset{(b)}{=}  <\bm{\Phi}^{(k)}, \lambda_{\bm{N},\bm{W}_i}^{(j)} \bm{\Psi}^{(i,j)} + \bm{o}^{(i,j)}>.
\end{aligned}
\end{equation}
The equation $(a)$ dues to eigenfunction expansion form. The equation $(b)$ owes to \eqref{eq:pseudo}.

Then 
\begin{equation}
\begin{aligned}    
   \| \bm{\Psi}^{(i,j)} -  \mathcal{P}_{\bm{\Phi}} \bm{\Psi}^{(i,j)}\|_F^2 &= \sum_{k = NMW(1+\epsilon)} ^{NM-1} <\bm{\Phi}^{(k)}, \bm{\Psi}^{(i,j)}>^2 \\
    &= \sum_{k = NMW(1+\epsilon)} ^{NM-1} \frac{<\bm{\Phi}^{(k)}, \bm{o}^{(i,j)}>^2}{|\lambda_{\bm{N},\bm{W}_i}^{(j)}-\lambda_{\bm{N},\mathbb{W}}^{(k)}|^2}.
\end{aligned}
\end{equation}

According to the bounds of $\lambda_{\bm{N},\bm{W}_i}^{(j)}$ in Theorem~\ref{thm:2Dcluster} and $\lambda_{\bm{N},\mathbb{W}}^{(k)}$ in Theorem~\ref{thm:cluster1}, the denominator is bounded as
\begin{equation}
\begin{aligned}
    &|\lambda_{\bm{N},\bm{W}_i}^{(j)}-\lambda_{\bm{N},\mathbb{W}}^{(k)}|^2 \ge  \left(1 - \tilde{C}_1(\mathbb{W},\epsilon) e^{-\tilde{C}_2(\mathbb{W},\epsilon) \min(M,N)}\right. \\ 
        &\left. \ -  \tilde{C}_3(\mathbb{W},\epsilon) e^{-\tilde{C}_4(\mathbb{W},\epsilon) \min(M,N)}  \right)^2 :=\delta_1\\
\end{aligned}
\end{equation}

According to the bound of $\|\bm{o}^{i,j}\|_F$ in \eqref{eq:obound}, we reach the final bound as
\begin{equation}\label{eq:projbound1}
\begin{aligned}
     &\| \bm{\Psi}^{(i,j)} -  \mathcal{P}_{\bm{\Phi}} \bm{\Psi}^{(i,j)}\|_F^2 \le \frac{\|\bm{o}^{i,j}\|_F^2}{\delta_1}\\
     &\le \frac{2\bar{C}_1(\bm{W},\epsilon)e^{-\bar{C}_2(\bm{W}_1,\epsilon)\min{(M,N)}}}{\delta_1}\\
     &\overset{(a)}{\le}  \frac{2\tilde{C}_1(\mathbb{W},\epsilon)e^{-\tilde{C}_2(\mathbb{W},\epsilon)\min{(M,N)}}}{\delta_1}.
\end{aligned}
\end{equation}
The inequality $(a)$ is due to the definitions of $\tilde{C}_1$ and $\tilde{C}_2$.

We use $\delta_2$ and $\delta_3$ to denote the upper bounds of $\|\bm{\Psi}^{(i,j)} -  \mathcal{P}_{\bm{\Phi}} \bm{\Psi}^{(i,j)}\|^2_F$ in \eqref{eq:projbound1} and $|<\|\bm{\Psi}^{(i,j)}, \bm{\Psi}^{(u,v)}\|>|$ in \eqref{eq:corbound1} for $(i,j) \neq (u,v)$.
Then for any $\bm{\psi} = \sum_{i \in [J]} \sum_{j < p_i} \alpha_{i,j}\bm{\Psi}^{(i,j)} \in \mathcal{S}_{\bm{\Psi}}$
\begin{equation}\label{eq:Fbound}
\begin{aligned}
    \|\bm{\psi}\|_F^2 &= \left\|\sum_{i,j} \alpha_{i,j}\bm{\Psi}^{(i,j)}\right\|_F^2 \\
    & \le \sum_{i,j} \alpha_{i,j}^2 \|\bm{\Psi}^{(i,j)}\|_F^2  +\sum_{i.j}\sum_{u,v,(u,v)\neq(i.j)} |\alpha_{i,j} \alpha_{u,v}| \delta_3 \\
    & \le \sum_{i,j} \alpha_{i,j}^2 + \sum_{i.j}\sum_{u,v,(u,v)\neq(i.j)} \frac{\alpha_{i,j}^2 + \alpha_{u,v}^2}{2} \delta_3 \\
    &\le \left(\sum_{i,j} \alpha_{i,j}^2\right)(1+(MN-1)\delta_3 ).
\end{aligned}
\end{equation}

and 
\begin{equation}\label{eq:projpsibound}
\begin{aligned}
    &\|\mathcal{P}_{\bm{\Phi}}\bm{\psi}\|_F^2 = \left\|\sum_{i,j} \mathcal{P}_{\bm{\Phi}}(\alpha_{i,j} \bm{\Psi}^{(i,j)})\right\|_F^2\\
    &=  \sum_{i,j} \alpha_{i,j}^2 \|\mathcal{P}_{\bm{\Phi}}\bm{\Psi}^{(i,j)}\|_F^2 + \\
    &\ \ \ \sum_{i,j}\sum_{u,v,(u,v)\neq (i,j)} <\alpha_{i,j} \mathcal{P}_{\bm{\Phi}} \bm{\Psi}^{(i,j)}, \alpha_{u,v} \mathcal{P}_{\bm{\Phi}} \bm{\Psi}^{(u,v)}>  \\
    &=  \sum_{i,j} \alpha_{i,j}^2 \|\mathcal{P}_{\bm{\Phi}}\bm{\Psi}^{(i,j)}\|_F^2 + \sum_{i,j}\sum_{u,v,(u,v)\neq (i,j)} \\
    &\ \ \ <\alpha_{i,j}\bm{\Psi}^{(i,j)},\alpha_{u,v} (\bm{\Psi}^{(u,v)}-(\bm{\Psi}^{(u,v)}-\mathcal{P}_{\bm{\Phi}} \bm{\Psi}^{(u,v)}))>  \\
    & \overset{(a)}{\ge}  \sum_{i,j} \alpha_{i,j}^2 (1-\delta_2) \\
    &\ \ - \sum_{i,j}\sum_{u,v,(u,v)\neq (i,j)} |\alpha_{i,j} \alpha_{u,v}| (\delta_3 +\sqrt{\delta_2}) \\
    &\overset{(b)}{\ge}  (\sum_{i,j} \alpha_{i,j}^2)(1-MN(\delta_3 + \sqrt{\delta_2})).
\end{aligned}
\end{equation}
The inequality $(a)$ is due to $<\bm{\Psi}^{(i,j)},\bm{\Psi}^{(u,v)}> \le \delta_3$ and $\|\bm{\Psi}^{(u,v)}-\mathcal{P}_{\bm{\Phi}} \bm{\Psi}^{(u,v)})\| \le \sqrt{\delta_2}$. The inequality $(b)$ follows the same procedures in \eqref{eq:Fbound}.

Then one has
\begin{equation}
     \cos^2(\Theta_{\mathcal{S}_{\bm{\Phi}} \mathcal{S}_{\bm{\Psi}} }) = \inf_{\bm{\psi} \in \mathcal{S}_{\bm{\Psi}}} \frac{\|\mathcal{P}_{\bm{\Phi}}\bm{\psi}\|_F^2}{\|\bm{\psi}\|_F^2} \ge \frac{1- MN(\delta_3 + \sqrt{\delta_2})}{1+MN\delta_3}.
\end{equation}

The original theorem is proved by substituting corresponding $\delta_2$ and $\delta_3$.

\section{Proof of Theorem~\ref{thm:pro2}}\label{proof:pro2}

Define a new operator $\bar{\mathcal{B}}(\cdot) = \sum_{i \in [J]}\sum_{j < q_i} \lambda^{(j)}_{\bm{N},\bm{W}_i} \bm{\Psi}^{(i,j)} <\bm{\Psi}^{(i,j)},\cdot>$.

Then, we have
\begin{equation}
\begin{aligned}
    &<\bm{y},\bar{\mathcal{B}}(\bm{y})> = \sum_{i \in [J]}\sum_{j < q_i}<\bm{y}, \lambda^{(j)}_{\bm{N},\bm{W}_i} \bm{\Psi}^{(i,j)} <\bm{\Psi}^{(i,j)},\bm{y}>> \\
    &= \sum_{i \in [J]}\sum_{j < q_i}<\bm{y},\bm{\Psi}^{(i,j)}> \lambda^{(j)}_{\bm{N},\bm{W}_i} <\bm{\Psi}^{(i,j)},\bm{y}>  \\
    &\le \sum_{i \in [J]}\sum_{j < MN} \lambda^{(i)}_{\bm{N},\bm{W}} |<\bm{\Psi}^{(i,j)},\bm{y}>|^2\\
    &= <\bm{y}\mathcal{B}_{\bm{N},\mathbb{W}}(\bm{y})>\\
    &= \int_{\mathbb{W}} |\tilde{\bm{y}}(\bm{f})|^2 \text{d}\bm{f},
\end{aligned}
\end{equation}
which yields $<\bm{y},\bar{\mathcal{B}}(\bm{y})>  \le \|\bm{y}\|_F^2$. Therefore, the eigenvalues of $\bar{\mathcal{B}}(\cdot)$ are not greater than 1 and $\bar{\mathcal{B}}$ is semi-positive. We can perform eigendecomposition on $\bar{\mathcal{B}}(\cdot)$ as 
$\bar{\mathcal{B}}(\cdot) = \sum_{k} \lambda^{(k)} \bar{\bm{\Psi}}_k <\bar{\bm{\Psi}}_k, \cdot>$. The space spanned by $ \bar{\bm{\Psi}}_k$ is the same as the space spanned by $\bm{\Psi}^{(i,j)}$.
We have
\begin{equation}
    \begin{aligned}
        \|\mathcal{P}_{\bm{\Psi}} \bm{y}\|_F^2 &= \|\sum_{k} \bar{\bm{\Psi}}_k <\bar{\bm{\Psi}}_k, \bm{y}>\|_F^2 \\
        &\ge\| \sum_{k} \lambda^{(k)} \bar{\bm{\Psi}}_k <\bar{\bm{\Psi}}_k, \bm{y}>\|_F^2 \\
        &= \|\bar{\mathcal{B}}(\bm{y})\|_F^2.
    \end{aligned}
\end{equation}

Consider back to the dictionary $\bm{\Phi}$, we obtain
\begin{equation}
\begin{small}
    \begin{aligned}
        &\|\mathcal{P}_{\bm{\Psi}} \bm{\Phi}^{(k)}\|_F \ge \|\bar{\mathcal{B}}(\bm{\Phi}^{(k)})\|_F \\
        &= \|{\mathcal{B}}_{\bm{N},\mathbb{W}}(\bm{\Phi}^{(k)}) - \sum_{i \in [J]} \sum_{j= q_i}^{MN-1} \lambda^{(j)}_{\bm{N},\bm{W}_i} \bm{\Psi}^{(i,j)} <\bm{\Psi}^{(i,j)},\bm{\Phi}_k> \|_F \\
        &\ge \|{\mathcal{B}}_{\bm{N},\mathbb{W}}(\bm{\Phi}^{(k)})\|_{F} - \sum_{i \in [J]} \sum_{j= q_i}^{MN-1} \|\lambda^{(j)}_{\bm{N},\bm{W}_i} \bm{\Psi}^{(i,j)} <\bm{\Psi}^{(i,j)},\bm{\Phi}_k> \|_F \\
        &\ge \lambda^{(k)}_{\bm{N},\mathbb{W}} -  \sum_{i \in [J]} \sum_{j= q_i}^{MN-1} \lambda^{(j)}_{\bm{N},\bm{W}_i} \\
        &\ge 1-6M^2 N^2 C_5(\mathbb{W},\epsilon)e^{-\tilde{C}_2(\mathbb{W},\epsilon) \min{(M,N)}} \\
        \ \ &-MN\tilde{C}_3(\mathbb{W},\epsilon)e^{-\tilde{C}_4(\mathbb{W},\epsilon)\min{(M,N)}}.
    \end{aligned}
    \end{small}
\end{equation}

For the ease of representation, we denote $\delta_4$ as
\begin{equation}
\begin{aligned}
    \delta_4 &:= 6M^2 N^2 C_5(\mathbb{W},\epsilon)e^{-\tilde{C}_2(\mathbb{W},\epsilon) \min{(M,N)}} \\
        &\ \ + MN\tilde{C}_3(\mathbb{W},\epsilon)e^{-\tilde{C}_4(\mathbb{W},\epsilon)\min{(M,N)}} \\
        &\ge 1 -  \|\mathcal{P}_{\bm{\Psi}} \bm{\Phi}^{(k)}\|_F.
\end{aligned}
\end{equation}

Then one can bound the difference between $\bm{\Phi}^{(k)}$ and its projection as
\begin{equation}
\begin{aligned}
    \ \ &\| \bm{\Phi}^{(k)} - \mathcal{P}_{\bm{\Psi}} \bm{\Phi}^{(k)}\|_F^2 \overset{(a)}{=} \| \bm{\Phi}^{(k)}\|_F^2 - \|\mathcal{P}_{\bm{\Psi}} \bm{\Phi}^{(k)}\|_F^2\\
    &\le 1-(1-\delta_4)^2 \le 2\delta_4.
\end{aligned}
\end{equation}
The equality $(a)$ is due to the Pythagorean Theorem.

Then for any $\bm{\phi} = \sum_{k} \alpha_k  \bm{\Phi}^{(k)} \in \mathcal{S}_{\Phi}$, we follow the same procedures in \eqref{eq:Fbound}
and \eqref{eq:projpsibound}, to bound $\|\bm{\phi}\|_F^2$ and $\mathcal{P}_{\bm{\Psi}} \bm{\phi}$. Noting that $<\bm{\Phi}^{(k)},<\bm{\Phi}^{(l)}> = 0$ for $k\neq l$, we have 
\begin{equation}
    \|\bm{\phi}\|_F^2 = (\sum_k \alpha_k^2),
\end{equation}
and
\begin{equation}
    \|\mathcal{P}_{\bm{\Psi}}\bm{\phi}\|_F^2 \ge (\sum_k \alpha_k^2)(1-MN \sqrt{2\delta_4}).
\end{equation}

Therefore, we complete the proof by
\begin{equation}
     \cos^2(\Theta_{\mathcal{S}_{\bm{\Psi}} \mathcal{S}_{\bm{\Phi}} }) = \inf_{\bm{\phi} \in \mathcal{S}_{\bm{\Phi}}} \frac{\|\mathcal{P}_{\bm{\Psi}}\bm{\phi}\|_F^2}{\|\bm{\phi}\|_F^2} \ge 1-MN\sqrt{2\delta_4}.
\end{equation}

\section{Proof of Theorem~\ref{thm:parasharp}}\label{proof:para}

In this proof, we consider the trace and Frobenius norm of $\mathcal{B}_{\bm{N},\mathbb{W}^{\Diamond}}$. Like the proof of Theorem~\ref{thm:sharp}, we bound the number of eigenvalues in $[\epsilon,1-\epsilon]$ by calculating $\sum_{p}\lambda_{\bm{N},\mathbb{W}^{\Diamond}}^{(p)}(1-\lambda_{\bm{N},\mathbb{W}^{\Diamond}}^{(p)})$.
Firstly, we consider the $i$-th region $\mathbb{W}^{\Diamond}_i$. We assume its center point is placed on the origin, as the location does not influence the eigenvalues of $\mathbb{W}^{\Diamond}_i$.
Given $ad-bc = V \neq 0$, the $i$-th  region on the frequency domain is 
\begin{equation}
\begin{aligned}
\mathbb{W}^{\Diamond}_i = \left\{(f,g): 
    -W^{[0]}_i\le af + bg \le W^{[0]}_i, \right.\\
   \left.-W^{[1]}_i\le cf + dg \le W^{[1]}_i.
   \right\}
   \end{aligned}
\end{equation}

The covariance of the samples at the ``time" domain $[m,n]$ and $[p,q]$ is given by
\begin{equation}
\begin{footnotesize}
\begin{aligned}
    &\ B_i (m,n,p,q) = \int_{\mathbb{W}^{\Diamond}_i} exp(j2\pi (m-p)f + j2\pi (n-q)g) \text{d}f \text{d}g \\
    &= \frac{1}{|V|} \int_{-W^{[1]}_i}^{W^{[1]}_i}\int_{-W^{[0]}_i}^{W^{[0]}_i} exp(j2\pi (m-p)(d/V u -b/V v)  \\
    &\ \ \ +  j2\pi(n-q)(-c/V u + a/V v)) \text{d}u\text{d}v \\
    &= \frac{1}{|V|} \int_{-W^{[1]}_i}^{W^{[1]}_i}\int_{-W^{[0]}_i}^{W^{[0]}_i} exp(j2\pi/V (md-pd-nc+qc)u) \\
    &\ \ \ exp(j2\pi/V (-mb+pb+na-qa)v)\text{d}u\text{d}v \\
    &= |V| \frac{\sin(2\pi/V W^{[0]}_i (-(n-q)c+(m-p)d))}{\pi(-(n-q)c+(m-p)d)} \\
    &\ \ \ \frac{\sin(2\pi/V W^{[1]}_i ((n-q)a-(m-p)b))}{\pi((n-q)a-(m-p)b))}
\end{aligned}
\end{footnotesize}
\end{equation}

Then we can calculate the traces of $\mathcal{B}_{\bm{N},\mathbb{W}^{\Diamond}_i}$ and $\mathcal{B}_{\bm{N},\mathbb{W}^{\Diamond}}$ respectively, as
\begin{equation}
trace(\mathcal{B}_{\bm{N},\mathbb{W}^{\Diamond}_i}) = 4MNW^{[0]}_i W^{[1]}_i /|V| = MN|\mathbb{W}^{\Diamond}_i|,
\end{equation}
and
\begin{equation}
trace(\mathcal{B}_{\bm{N},\mathbb{W}}^{\Diamond}) = \sum_{i} trace(\mathcal{B}_{\bm{N},\mathbb{W}^{\Diamond}_i}) = MN|\mathbb{W}^{\Diamond}|.
\end{equation}

Take $s = n-q$ and $t = m-p$. The square of the Frobenius norm of $\mathcal{B}_{\bm{N},\mathbb{W}^{\Diamond}_i}$ is calculated as
\begin{equation}
    \begin{aligned}
        F_i = &\sum_{n,q =0}^{N-1} \sum_{m,p = 0}^{M-1} B(m,n,p,q)^2\\
        & =|V|^2 \sum_{s = 1-N}^{N-1} (N-|s|) \sum_{t = 1-M}^{M-1} (M-|t|) \\
        &\ \ \frac{\sin^2(2\pi/V(-sc+td))}{\pi^2(-sc+td)^2} \frac{\sin^2(2\pi/V(sa-tb))}{\pi^2(sa-tb)^2}.
    \end{aligned}
\end{equation}

Below, we separate $F$ into two parts $F_a$ and $F_b$
\begin{equation}
\begin{aligned}
     &F_a  =|V|^2MN\sum_{s = 1-N}^{N-1} \sum_{t = 1-M}^{M-1} \\
     &\frac{\sin^2(2\pi/V(-sc+td))}{\pi^2(-sc+td)^2} \frac{\sin^2(2\pi/V(sa-tb))}{\pi^2(sa-tb)^2}.
\end{aligned}
\end{equation}

\begin{equation}
\begin{footnotesize}
\begin{aligned}
    &F_b =  \\
    &|V|^2 \sum_{s = 1-N}^{N-1} |s| \sum_{t = 1-M}^{M-1} M \frac{\sin^2(2\pi/V(-sc+td))}{\pi^2(-sc+td)^2} \frac{\sin^2(2\pi/V(sa-tb))}{\pi^2(sa-tb)^2}\\
    &+|V|^2 \sum_{s = 1-N}^{N-1} N \sum_{t = 1-M}^{M-1} |t| \frac{\sin^2(2\pi/V(-sc+td))}{\pi^2(-sc+td)^2} \frac{\sin^2(2\pi/V(sa-tb))}{\pi^2(sa-tb)^2} \\.
\end{aligned}
\end{footnotesize}
\end{equation}

One can verify that $F_i \ge F_a - F_b$. We then bound $F_a = MN|\mathbb{W}^{\Diamond}_i| - O(M+N)$ in Appendix~\ref{ap:subfa} and $F_b = O(M\log(N)+N\log(M))$ in Appendix~\ref{ap:subfb}.

With these bounds, we finally show the following inequality
\begin{equation}
\begin{aligned}
    &\ \ \sum_{p}\lambda^{(p)}_{\bm{N},\mathbb{W}^{\Diamond}}(1-\lambda^{(p)}_{\bm{N},\mathbb{W}^{\Diamond}}) \\
    &= trace(\mathcal{B}_{\bm{N},\mathbb{W}^{\Diamond}}) - \|\mathcal{B}_{\bm{N},\mathbb{W}^{\Diamond}}\|_{F}^2 \\
    &\le MN|\mathbb{W}^{\Diamond}| - \sum_{i=0}^{J-1} \|\mathcal{B}_{\bm{N},\mathbb{W}^{\Diamond}_i}\|_{F}^2\\
    & = O(JN\log(M) + JM\log(N)).
\end{aligned}
\end{equation}
The number of the eigenvalues lie in $[\epsilon,1-\epsilon]$ will be $O(JN\log(M) + JM\log(N))/(\epsilon(1-\epsilon))$, which is small compared to the total number $MN$.

\subsection{Bounding $F_a$}\label{ap:subfa}
Note that $ad-bc \neq 0$. Therefore, $-sc+td$ and $sa-tb$ are both equal to 0 only when $s= t = 0$. 

\begin{equation}
\begin{footnotesize}
\begin{aligned}
    &F_a  \\
    &=|V|^2MN\sum_{s = 1-N}^{N-1} \sum_{t = 1-M}^{M-1} \frac{\sin^2(2\pi/V(-sc+td))}{\pi^2(-sc+td)^2} \frac{\sin^2(2\pi/V(sa-tb))}{\pi^2(sa-tb)^2}\\
    & = 4MNW_i^{[0]}W_i^{[1]}/|V| - \\
    &\sum_{\{s,t:|s|\ge N or |t| \ge M\}}\frac{\sin^2(2\pi/V(-sc+td))}{\pi^2(-sc+td)^2} \frac{\sin^2(2\pi/V(sa-tb))}{\pi^2(sa-tb)^2} \\
    & \ge 4MNW_i^{[0]}W_i^{[1]}/|V|  \\
    & -2MN|V|^2 \underbrace{\sum_{\{s,t: t\ge M\}}\frac{\sin^2(2\pi/V(-sc+td))}{\pi^2(-sc+td)^2} \frac{\sin^2(2\pi/V(sa-tb))}{\pi^2(sa-tb)^2}}_{\alpha} \\
    &-2MN|V|^2\underbrace{\sum_{\{s,t:s\ge N\}}\frac{\sin^2(2\pi/V(-sc+td))}{\pi^2(-sc+td)^2} \frac{\sin^2(2\pi/V(sa-tb))}{\pi^2(sa-tb)^2}}_{\beta} 
\end{aligned}
\end{footnotesize}
\end{equation}

Below, we bound $\alpha$. Denote $s_1 = \lfloor\frac{td}{c}\rfloor$ and $s_2 = \lfloor\frac{tb}{a}\rfloor$. We assume $s_2 - s_1 > 4$ without loss of generality when $M$ is sufficiently large and $t \ge M$.
\begin{equation}
\begin{aligned}
\alpha &\le \sum_{t \ge M} \sum_{s=-\infty}^{s_1 -2} \frac{1}{\pi^2(-cs+dt)^2} \frac{1}{\pi^2(as-bt)^2} \\
&+ \sum_{t \ge M} \sum_{s=s_1 +2}^{s_2 -2} \frac{1}{\pi^2(-cs+dt)^2} \frac{1}{\pi^2(as-bt)^2} \\
&+ \sum_{t \ge M} \sum_{s=s_2 +2}^{\infty} \frac{1}{\pi^2(-cs+dt)^2} \frac{1}{\pi^2(as-bt)^2}\\
&+ \sum_{t \ge M}\sum_{s=s_1 -1}^{s_1+1}\frac{4}{V^2}\frac{1}{\pi^2 (as-bt)^2}\\
& +\sum_{t \ge M}\sum_{s=s_2 -1}^{s_2+1}\frac{4}{V^2}\frac{1}{\pi^2 (cs-dt)^2}
\end{aligned}
\end{equation}

In the above inequality, we separate the region $\{(s,t)\in \mathbb{Z}^2: t\ge M\}$ into several disjoint parts. Then we try to give the bounds of these summations with integrals over these continuous regions. The first three parts are $\{(s,t) \in \mathbb{Z}^2: t\ge M, s \le s_1 -2\}$, $\{(s,t)\in \mathbb{Z}^2: t\ge M, s_1 + 2\le s \le s_2 -2\}$ and $\{(s,t)\in \mathbb{Z}^2: t\ge M, s \ge s_2 +2\}$. We take $ Z =  \{(s,t)\in \mathbb{Z}^2: t\ge M, s_1 + 2\le s \le s_2 -2\}$ for example. For any $(s,t) \in Z$ and a large $M$, at least one region $C^*$ among $[s,s+1]\times[t,t+1]$, $[s-1,s]\times[t,t+1]$,$[s,s+1]\times[t-1,t]$ and $[s-1,s]\times[t-1,t]$ is a subset of $C = \{(x,y) \in \mathbb{R}^2: y\ge M-1, td/c + 1\le x \le tb/a -1\}$. We can bound  $\frac{1}{(-cs+dt)^2} \frac{1}{(as-bt)^2}$ with an integral on $C^*$ as
\begin{equation}
\begin{aligned}
   &\ \  \frac{1}{(-cs+dt)^2} \frac{1}{(as-bt)^2}\\
   &=  \frac{1}{(-cs+dt)^2} \frac{1}{(as-bt)^2} \frac{\int_{C^*} \frac{1}{(-cx+dy)^2} \frac{1}{(ax-by)^2}  dxdy}{\int_{C^*} \frac{1}{(-cx+dy)^2} \frac{1}{(ax-by)^2}  \text{d}x\text{d}y}\\
    &\le \frac{(|-cs+dt|+|c|+|d|)^2}{(-cs+dt)^2} \frac{(|as-bt|+|a|+|b|)^2}{(as-bt)^2} \\
    &\ \ \ \ \int_{C^*} \frac{1}{(-cx+dy)^2} \frac{1}{(ax-by)^2}  \text{d}x\text{d}y \\
    &\le \frac{(2|c|+|d|)^2(2|a|+|b|)^2}{a^2c^2} \int_{C^*} \frac{1}{(-cx+dy)^2} \frac{1}{(ax-by)^2}  \text{d}x\text{d}y.
\end{aligned}
\end{equation}
The first inequality is due to $\frac{1}{(-cx+dy)^2} \frac{1}{(ax-by)^2} \le \frac{1}{(|-cs+dt|+|c|+|d|)^2} \frac{1}{|as-bt|+|a|+|b|)^2}$ for any $(x,y) \in C^*$. The second inequality is due to $|-cs+dt| \ge c$ and $|-as+bt| \ge a$ for any $(s,t) \in Z$. For the other two parts $\{(s,t) \in \mathbb{Z}^2: t\ge M, s \le s_1 -2\}$ and $\{(s,t)\in \mathbb{Z}^2: t\ge M, s \ge s_2 +2\}$, the same bound holds.

The last two parts are  $\{(s,t)\in \mathbb{Z}^2: t\ge M, s_1 - 1\le s \le s_1 + 1\}$ and $\{(s,t)\in \mathbb{Z}^2: t\ge M, s_2 - 1\le s \le s_2 + 1\}$. Similarly, we achieve the following two bounds for $(s,t)$ in each part.
\begin{equation}
\begin{aligned}
    \frac{1}{(as-bt)^2} &= \frac{1}{(as-bt)^2} \frac{\int_{t}^{t+1} \frac{1}{(ax-by)^2}\text{d}y}{\int_{t}^{t+1} \frac{1}{(ax-by)^2}\text{d}y}\\
    &\le\frac{|as-bt|+|a|+|b|}{(as-bt)^2} \int_{t}^{t+1} \frac{1}{(ax-by)^2}\text{d}y\\
    &\le
    \frac{(2|a|+|b|)^2}{a^2} \int_t^{t+1}\frac{1}{(ax-by)^2}\text{d}y
\end{aligned}
\end{equation}
\begin{equation}
     \frac{1}{(-cs+dt)^2} \le \frac{(2|c|+|d|)^2}{c^2} \int_t^{t+1}\frac{1}{(-cx+dy)^2}\text{d}y
\end{equation}

Take the constant 
\begin{equation}
\begin{aligned}
    k &= 4 \max\left(\frac{(2|c|+|d|)^2(2|a|+|b|)^2}{a^2c^2},\right.\\
    &\left.\frac{(2|a|+|b|)^2}{a^2},\frac{(2|c|+|d|)^2}{c^2}\right),
\end{aligned}
\end{equation}
noting that the coefficient 4 is because the $1\times1$ region may be used by multiple (s,t) grids, and we can transfer the summation into integral form as
\begin{equation}
\begin{aligned}
    \alpha  &\le k\int_{t=M-1}^{\infty} \int_{s=-\infty}^{td/c -1} \frac{1}{\pi^2(-cs+dt)^2} \frac{1}{\pi^2(as-bt)^2} \text{d}s\text{d}t \\
&+ k\int_{t=M-1}^{\infty} \int_{s=td/c+1}^{tb/a - 1} \frac{1}{\pi^2(-cs+dt)^2} \frac{1}{\pi^2(as-bt)^2} dsdt\\
&+ k\int_{t=M-1}^{\infty} \int_{s=tb/a -1}^{\infty} \frac{1}{\pi^2(-cs+dt)^2} \frac{1}{\pi^2(as-bt)^2} \text{d}s\text{d}t\\
&+ k\int_{t=M-1}^{\infty}\sum_{s=s_1 -1}^{s_1+1}\frac{4}{V^2}\frac{1}{\pi^2 (as-bt)^2} \text{d}t\\
& +k\int_{t=M-1}^{\infty}\sum_{s=s_2 -1}^{s_2+1}\frac{4}{V^2}\frac{1}{\pi^2 (cs-dt)^2} \text{d}t.
\end{aligned}
\end{equation}

Note that 
\begin{equation}
\begin{aligned}
    &\ \ \int \frac{1}{(-cs+dt)^2} \frac{1}{(as-bt)^2}ds \\
    &= \frac{-2acs+adt+bct}{t^2V^2(as-bt)(cs-dt)} + \frac{2ac}{t^3V^3}\log(\frac{|as-bt|}{|cs-dt|}).
\end{aligned}
\end{equation}
One can verify that all the items on the left of the second inequality are $O(\frac{1}{M})$. Consequently, we have $\alpha = O(\frac{1}{M})$, $\beta = O(\frac{1}{N})$ and $F_a =  4MNW_i^{[0]}W_i^{[1]}/|V| - O(M+N)$.

\subsection{Bounding $F_b$}\label{ap:subfb}
The process of bounding $F_b$ is much similar to that of $F_a$.
\begin{equation}
\begin{footnotesize}
\begin{aligned}
    F_b &=  |V|^2 \underbrace{\sum_{s = 1-N}^{N-1} |s| \sum_{t = 1-M}^{M-1} M \frac{\sin^2(2\pi/V(-sc+td))}{\pi^2(-sc+td)^2} \frac{\sin^2(2\pi/V(sa-tb))}{\pi^2(sa-tb)^2}}_{\alpha}\\
    &+|V|^2 \underbrace{\sum_{s = 1-N}^{N-1} N \sum_{t = 1-M}^{M-1} |t| \frac{\sin^2(2\pi/V(-sc+td))}{\pi^2(-sc+td)^2} \frac{\sin^2(2\pi/V(sa-tb))}{\pi^2(sa-tb)^2}}_{\beta} \\.
\end{aligned}
\end{footnotesize}
\end{equation}

Similarly, we can bound $\alpha$ with an integral on the region of interest and additional values. By omitting the intermediate steps and abusing some denotations, we obtain
\begin{equation}
\begin{aligned}
    \alpha &\le V^2 Nk\sum_{t = 1-M}^{M-1}|t| \frac{-2acs+adt+bct}{t^2V^2(as-bt)(cs-dt)} \\
    &\ \ + \frac{2ac}{t^3V^3}\log(\frac{|as-bt|}{|cs-dt|}) \Bigg|^{s=N}_{s=-N}  \\&
    \ \ +  V^2 Nk \sum_{t = 1-M}^{M-1} \frac{4}{V^2} \frac{1}{\pi^2(\frac{ad}{c}-b)^2 |t|}.
\end{aligned}
\end{equation}
One can verify that $\alpha = O(N\log(M))$ and similarly $\beta = O(M\log(N))$.

\bibliographystyle{ieeetr}
\bibliography{dpss}
\end{document}